\documentclass[11pt]{article}

\input{preamble}

\bibliography{papers}

\newcommand{\cluster}{\pname{$k$\hy{}Clustering$^q$}\xspace}
\newcommand{\fl}{\pname{Facility Location$^q$}\xspace}

\newcommand{\hw}{h}
\newcommand{\offset}{\tau(\eps,q,\sigma)}
\DeclareMathOperator{\cost}{cost}

\DeclareMathOperator{\dist}{dist}

\DeclareMathOperator{\diam}{diam}
\newcommand{\opt}{\ensuremath\textsc{OPT}}

\title{Polynomial Time Approximation Schemes for Clustering in Low Highway 
Dimension Graphs\footnote{A preliminary version of this paper appeared at ESA 
2020~\cite{DBLP:conf/esa/FeldmannS20}.}}

\author{Andreas Emil Feldmann\footnote{
Supported by the Czech Science Foundation GA{\v C}R
(grant \#19-27871X), and by the Center for Foundations of Modern Computer
Science (Charles Univ.\ project UNCE/SCI/004).
}\\ 
Charles University, Prague, Czechia \\ \texttt{feldmann.a.e@gmail.com}
\and
David Saulpic \\
LIP6, Sorbonne Universit\'e, Paris, France\\
\texttt{david.saulpic@lip6.fr}
}

\date{}

\begin{document}

\maketitle

\begin{abstract}
We study clustering problems such as $k$-Median, $k$-Means, and Facility 
Location in graphs of low highway dimension, which is a graph parameter modeling 
transportation networks. It was previously shown that approximation schemes for 
these problems exist, which either run in quasi-polynomial time (assuming 
constant highway dimension) [Feldmann et al.\ SICOMP 2018] or run in FPT time 
(parameterized by the number of clusters $k$, the highway dimension, and the 
approximation factor) [Becker et al.\ ESA~2018, Braverman et al.\ SODA 2021]. In 
this paper we show that a polynomial-time approximation scheme (PTAS) exists 
(assuming constant highway dimension). We also show that the considered problems 
are NP-hard on graphs of highway dimension 1.
\end{abstract}

\newpage
\setcounter{page}{1}

\section{Introduction}
Clustering is a standard optimization task that seeks a ``good'' partition of a 
metric space, such that two points that are ``close'' should be in the same 
part. A good clustering of a dataset allows to retrieve and exploit data, and 
is therefore a common routine in data analysis. The underlying data can come 
from various sources and represent many different objects. In particular, it is 
often interesting to cluster geographic data. In that case, the metric space can 
be given by a transportation network, which can be modeled  by graphs with low 
\emph{highway dimension} (see \cref{dfn:hd}).

In this article, we study some popular clustering objectives, namely 
\pname{Facility Location}, \pname{$k$\hy{}Median}, and \pname{$k$-Means}, in 
graphs with constant highway dimension. The two latter problems seek to find a 
set $S$ of $k$ points called \emph{centers} in a metric $(V, \dist)$ that 
minimizes the function $\sum_{v \in V} \min_{f \in S} \dist(v, f)^q$, with $q=1$ 
for \pname{$k$-Median} and $q=2$ for \pname{$k$-Means}. The objective for 
\pname{Facility Location} is slightly different: each point $f$ of the metric 
space has an \emph{opening cost}~$w_f$, and the goal is to find a set $S$ that 
minimizes $\sum_{f \in S} w_f + \sum_{v \in V} \min_{f \in S} \dist(v, f)$. 
These problems are APX-hard in general metric spaces, see for example 
\citet{guha1999greedy} for \pname{Facility Location}, \citet{jain2002new} for 
\pname{$k$-Median}, and \citet{AwasthiCKS15} for \pname{$k$-Means}.

To bypass the hardness of approximation known for these problems, researchers 
have considered low dimensional input, such as Euclidean spaces of fixed 
dimension, metrics with bounded doubling dimension, or metrics arising from 
classes of minor-free graphs. Many algorithmic tools were developed for that 
purpose: in their seminal work, \citet{AroraRagRao98} gave the first polynomial 
time approximation scheme (PTAS) for Euclidean \pname{$k$-Median} in 
$\mathbb{R}^2$, which generalizes to a quasi-polynomial time approximation 
scheme (QPTAS) in $\mathbb{R}^d$ for fixed $d$. This result was generalized by 
\citet{talwar2004bypassing}, who gave a QPTAS for metrics with bounded doubling 
dimension, and more recently by \citet{DBLP:conf/focs/SaulpicCF19}, who gave a 
near-linear time approximation scheme for this setting. 

In this work we focus on transportation networks, for which it can be argued 
that  metric spaces with bounded doubling dimension are not a suitable model: 
for instance, hub-and-spoke networks seen in air traffic networks do not have 
low doubling dimension (cf.~\cref{ap:hwdef}). Therefore we study graphs with 
constant \emph{highway dimension}, which formalize structural properties of such 
networks. The following definition is taken from \citet{FeldmannFKP15}. Here the 
\emph{ball}~$\beta_v(r)$ of radius $r$ around $v\in V$ is the set of all 
vertices at distance at most $r$ from $v$.

\begin{definition}
\label{dfn:hd}
The \emph{highway dimension} of a graph $G$ is the smallest integer~$\hw$ such 
that, for some universal constant $c>4$, for every $r\in \mathbb{R}^+$ and $v\in 
V$ there are at most $\hw$ vertices in the ball~$\beta_v(cr)$ of radius $cr$ 
around $v$ hitting all shortest paths of length more than $r$ that lie in 
$\beta_v(cr)$.
\end{definition}

Before this work, for this class of graphs the only known approximation 
algorithms for clustering that compute $(1+\eps)$-approximations for any 
$\eps>0$ either run in quasi-polynomial $n^{(\log n)^{O(\log^2(\hw/\eps))}}$ 
time, i.e., QPTASs~\cite{FeldmannFKP15}, or in $2^{\tilde 
O(\hw^{O(1/\eps)}+k)}\cdot n^{O(1)}$ time, i.e., parameterized approximation 
schemes~\cite{becker2017polynomial,braverman2020coresets}. Thus an open problem 
was to identify polynomial-time approximation schemes (PTASs) for clustering in 
graphs of constant highway dimension.

\subsection{Our results}

Our main result is a PTAS for clustering problems on graphs of 
constant highway dimension. 
For convenience, we define slightly more general problems than those stated 
above.
The \cluster problem is defined as follows. An instance $\mc{I}$ consists of a 
metric~$(V,\dist)$, a set of \emph{facilities} (or \emph{centers}) $F\subseteq 
V$, and a \emph{demand function} $\chi:V\to\mathbb{N}_0$. The goal is to find a 
set $S\subseteq F$ with $|S|\leq k$ minimizing  $\sum_{v \in V} \chi(v) \cdot 
\min_{f \in S} \dist(v,f)^q$ where $q$ is a positive integer. We call all 
vertices $v\in V$ with $\chi(v)>0$ the \emph{clients} of~$\mc{I}$. 
\pname{$k$-Median} and \pname{$k$-Means} are special cases of \cluster, where 
$q=1$ and $q=2$, respectively, and also $\chi(v)=1$ for all vertices $v\in V$.

The input to the \fl problem  is the same as for \cluster, but additionally each 
facility $f\in F$ has an \emph{opening cost} $w_f\in\mathbb{R}^+$. The goal is 
to find a set $S\subseteq F$ minimizing $\sum_{f\in S}w_f + \sum_{v \in V} 
\chi(v) \cdot \min_{f \in S} \dist(v,f)^q$ where $q$ is a positive integer. 
\pname{Facility Location} is a special case of \fl, where $q=1$ and $\chi(v)=1$ 
for all vertices $v\in V$.

Our main theorem\footnote{We remark that the success probability of the 
algorithm can be boosted to $1-\eps^\delta$ for any $\delta\in\mathbb{N}$ by 
running the algorithm $\delta$ times and outputting the best solution.} is the 
following, where $X=\max_{v\in V}\chi(v)$ is the largest demand (note that for 
\pname{$k$-Median}, \pname{$k$-Means}, or \pname{Facility Location} we 
have~$X=1$).

\begin{restatable}{theorem}{thmmainalg}
\label{thm:main-alg}
For any $\eps>0$, with probability $1-\eps$ a $(1+\eps)$-approximation for 
\cluster and \fl can be computed in $(nX)^{(\hw q/\eps )^{O(q)}}$ time on graphs 
of highway dimension~$\hw$ with $n$ vertices.
\end{restatable}

In particular, this algorithm is much faster than the quasi-polynomial time 
approximation scheme of \citet{FeldmannFKP15}
for \pname{$k$-Median} or \pname{Facility Location}. The runtime of our 
algorithm also significantly improves over the exponential dependence on $k$ in 
the approximation schemes of \citet{becker2017polynomial, braverman2020coresets} for \pname{$k$-Median}.

It has so far been open whether these clustering problems are NP-hard on graphs 
of constant highway dimension. We complement our main theorem by showing that 
they are NP-hard even for the smallest possible highway dimension. This answers 
an open problem given in~\cite{FeldmannFKP15}. Here the \emph{uniform} \fl 
problem has unit opening costs for all facilities.

\begin{restatable}{theorem}{thmhardness}
\label{thm:hardness}
The \cluster and uniform \fl problems are NP\hy{}hard on graphs of highway 
dimension $1$.
\end{restatable}

\subsection{Related work}

\subparagraph*{On clustering problems.} The problems we focus on in this article 
are known to be APX\hy{}hard in general metric spaces (see 
e.g.~\cite{guha1999greedy,jain2002new,AwasthiCKS15}). The current best 
polynomial-time algorithm for \pname{Facility Location} achieves a 
$1.488$-approximation~\cite{Li13}, while  the best approximation factor is 
$2.67$ for \pname{$k$-Median}~\cite{ByrkaPRST15} and $6.357$ for 
\pname{$k$-Means}~\cite{ahmadian2019better}. 

When restricting the inputs, a near-linear time approximation scheme for 
doubling metrics was developed in \cite{DBLP:conf/focs/SaulpicCF19}; we will 
discuss the close relations between our work and this one in~\cref{sec:ourtec}. 
Local search techniques also yield a PTAS in metrics arising from classes of 
minor-free graphs and metrics with bounded doubling dimension 
\cite{cohen2019local, FriggstadRS19}, and a $\Theta(q)$-approximation for the 
\cluster problem in general metric spaces \cite{abs-0809-2554}.

Another technique for dealing with clustering problems is to compute 
\emph{coresets}, which are compressed representations of the input. An 
$\eps$-coreset is a weighted set of points such that for every set of centers, 
the cost for the original set of points is within a $(1+\eps)$-factor of the 
cost for the coreset. \citet{braverman2020coresets} recently proved that graphs 
with highway dimension~$\hw$ admit coreset of size $\widetilde 
O((k+\hw)^{O(1/\eps)})$. This enables to compute a $(1+\eps)$-approximation by 
enumerating all possible solutions of the coreset. However, this coreset does 
not have small highway dimension,\footnote{Indeed, a subset of a metric with 
small highway dimension does not necessarily have small highway dimension as 
well: think of a star metric on which the center is removed.} and thus cannot be 
used to boost our algorithms. 

\subparagraph*{On highway dimension.}
The highway dimension was originally defined by \citet{abraham2010highway}, who 
specifically chose balls of radius $4r$ in the \cref{dfn:hd}. Since the original 
definition in~\cite{abraham2010highway}, several other definitions have been 
proposed. In particular, \citet{FeldmannFKP15} proved that when choosing 
a radius $cr$ in \cref{dfn:hd} for any constant $c$ strictly larger than~$4$, it 
is possible to exploit the structure of graphs with constant highway dimension 
in order to obtain a QPTAS for problems such as TSP, {\sc Facility Location}, 
and {\sc Steiner Tree}. As \citet{abraham2010highway} point out, the choice 
of the constant is somewhat arbitrary, and we use the above definition so 
that we may exploit the structural insights of~\cite{FeldmannFKP15} for our 
algorithm. These structural properties were also leveraged by 
\citet{becker2017polynomial} who gave a PTAS for the 
\pname{Bounded-Capacity Vehicle Routing} problem, and a parameterized 
approximation scheme for the \pname{$k$-Center} problem and \pname{$k$-Median}. 
In previous work, \citet{DBLP:conf/icalp/Feldmann15} gave a parameterized 
$3/2$-approximation algorithm with runtime~$2^{O(k\hw\log\hw)}n^{O(1)}$ for 
\pname{$k$-Center}. \citet{DBLP:conf/wg/DisserFKK19} showed that \pname{Steiner 
Tree} and \pname{TSP} are weakly NP-hard when the highway dimension is $1$, 
i.e., each of them is NP-hard but an FPTAS exists.

It is worth mentioning that further definitions of the highway dimension exist 
(for a detailed discussion see \cref{ap:hwdef} 
and~\cite{blum2019hierarchy,FeldmannFKP15}). 

\subsection{Our techniques}\label{sec:ourtec}

To obtain \cref{thm:main-alg}, we rely on the framework recently developed by
\citet{DBLP:conf/focs/SaulpicCF19} for doubling metrics. They show that the 
\emph{split-tree decomposition} of \citet{talwar2004bypassing} has some 
interesting properties, and exploit them to design their algorithm. Our main 
contribution is to provide a decomposition with similar properties in graphs 
with constant highway dimension. This is done relying on some structural 
properties of such graphs presented by \citet{FeldmannFKP15}.
We start by giving an outline of the algorithm 
from~\cite{DBLP:conf/focs/SaulpicCF19}, and then explain how to carry the 
results over to the highway dimension setting.

\subparagraph*{On doubling metrics.} 
The starting point of many approximation algorithms for doubling metrics is a 
decomposition of the metric, as presented in the following lemma taken
from~\cite{DBLP:conf/focs/SaulpicCF19}. A \emph{hierarchical decomposition} 
$\mc{D}$ of a metric $(V,\dist)$ is a sequence of partitions $\mc{B}_0, 
\mc{B}_1, \ldots, \mc{B}_\lambda$ of $V$, where $\mc{B}_{i}$ 
refines~$\mc{B}_{i+1}$, i.e., every part $B\in\mc{B}_i$ is contained in some 
part of $\mc{B}_{i+1}$. Moreover, in $\mc{B}_0$ every part contains a singleton 
vertex of $V$, while $\mc{B}_\lambda$ contains only one part, namely $V$. For a 
point $v\in V$ and a radius $r >0$, we say that the ball $\beta_v(r)$ is 
\emph{cut at level $i$} if $i$ is the largest integer for which the ball 
$\beta_v(r)$ is not contained in a single part of $\mc{B}_i$. For any subset 
$W\subseteq V$ of vertices we define $\lambda(W)=\lceil\log_2\diam(W)\rceil$, 
where $\diam(W)=\max_{u,v\in W}\dist(u,v)$ is the diameter of $W$.

\begin{lemma}[Reformulation of \cite{talwar2004bypassing, 
BartalG13} as found in \cite{DBLP:conf/focs/SaulpicCF19}\footnote{We remark that 
in~\cite{DBLP:conf/focs/SaulpicCF19} the preciseness of \cref{lem:talwar-decomp} 
was expressed akin to the weaker property found in \cref{lem:decomp}, which 
however would not lead to a near-linear time approximation scheme as claimed 
in~\cite{DBLP:conf/focs/SaulpicCF19}, but rather a PTAS as shown in this work. 
This can however easily be alleviated for~\cite{DBLP:conf/focs/SaulpicCF19} by 
using the stronger preciseness as stated here in \cref{lem:talwar-decomp}.}] 
\label{lem:talwar-decomp}
For any metric $(V,\dist)$ of doubling dimension~$d$ and any $\rho > 0$, there 
exists a polynomial-time computable randomized hierarchical decomposition 
$\mc{D}=\{\mc{B}_0,\ldots,\mc{B}_{\lambda(V)}\}$ such that the diameter of each 
part $B \in \mc{B}_i$ is at most $2^{i+1}$, and:
\begin{enumerate}
  \item\textbf{Scaling probability}: for any 
$v\in V$, radius $r$, and level $i$, we have \\
$\Pr[\mc{D}\text{ cuts } \beta_v(r)\text{ at level } i]\leq 2^{O(d)}\cdot 
r/2^i$. 

  \item\textbf{Portal set:}
    every part $B\in \mc{B}_i$ where $\mc{B}_i\in\mc{D}$ comes with a set of
    \emph{portals}~$P_B\subseteq B$ that is
    \begin{enumerate}
    \item \textbf{concise:} the size of the portal set is bounded by 
$|P_{B}| \le 1/\rho^d$, and

    \item \textbf{precise:} for every node $u\in B$ there is a portal 
$p\in P_B$ with $\dist(u,p)\leq \rho 2^{i+1}$. %

    \end{enumerate}
  \end{enumerate}
\end{lemma}

We briefly sketch the standard use of this decomposition. For clustering 
problems, one can show that there exists a \emph{portal-respecting solution} 
with near-optimal cost (see~\citet{talwar2004bypassing}). In this structured 
solution, each client connects to a facility via a \emph{portal-respecting path} 
that enters and leaves any part $B$ of $\mc{D}$ only through a node of the 
portal set $P_B$. Those portals therefore act as separators of the metric. A 
standard dynamic program approach can then compute the best portal respecting 
solution.

To ensure that there is a portal-respecting solution with near-optimal cost, one 
uses the preciseness property of the portal set: the distortion (i.e., the 
overhead) of connecting a client $c$ with a facility~$f$ through portals instead 
of directly, is bounded as follows. Let $i$ be the level at which \emph{$\mc{D}$ 
cuts} $c$ and~$f$, meaning that $i$ is the maximum integer for which $c$ and $f$ 
lie in different parts of~$\mc{B}_i$. At every level $j \leq i$ the 
portal-respecting path uses an edge to the closest portal on this level, and 
thus incurs a distortion of $\rho 2^{j+1}$. Hence the total distortion is 
$\sum_{j \leq i} O(\rho 2^j) = O(\rho 2^i)$. Now, the scaling probability of the 
decomposition ensures that $c$ and~$f$ are cut at level $i$ with probability at 
most $2^{O(d)}\dist(c, f) / 2^i$. Hence combining those two bounds over all 
levels ensures that, in expectation, the distortion between $c$ and $f$ is 
$2^{O(d)}\dist(c,f) \cdot \rho \lambda(V)$. Using a standard pre-processing 
technique (see~e.g.~\cite{FeldmannFKP15}) we may reduce $\lambda(V)$ to 
$O(\log(n/\eps))$ when aiming for a $(1+\eps)$-approximation. Hence choosing 
$\rho = \frac{\eps}{2^{O(d)} \log n}$ gives a distortion of 
$\eps\cdot\dist(c,f)$. Summing over all clients proves that there exists a 
near-optimal portal-respecting solution.

The issue with this approach is that by the conciseness property, the number of 
needed portals is $2^{O(d^2)}\log^d n/\eps^d$, and the dynamic program has a 
runtime that is exponential in this number. Thus the time complexity is 
quasipolynomial. The novelty of \citep{DBLP:conf/focs/SaulpicCF19} is to show 
how to reduce the number of portals to a constant. The idea is to reduce the 
number of levels on which a client can be cut from its facility.

For this, \citet{DBLP:conf/focs/SaulpicCF19} present a processing step of the 
instance that helps to deal with clients cut from their facility at a high level 
(see \cref{sec:prelim} for formal definitions and lemmas). Roughly speaking, 
their algorithm computes a constant factor approximation $L$ of \cluster or \fl, 
and a client $c$ is called \emph{badly-cut} if $\mc{D}$ cuts it from its closest 
facility of~$L$ at a level larger than $\log(\dist(c, L)/\eps)+\tau(\eps,q,d)$ 
for some function $\tau$. Every badly-cut client is moved to its closest 
facility of~$L$. It is then shown that this new instance $\mc{I}_\mc{D}$ has 
\emph{small distortion}, which essentially means that any solution to 
$\mc{I}_\mc{D}$ can be converted to a solution of the original instance~$\mc{I}$ 
while only losing a $(1+\eps)$-factor in quality. In this 
instance~$\mc{I}_\mc{D}$ all clients are cut from their closest facility of $L$ 
at some level between $\log (\dist(c, L)/2)$ and $\log (\dist(c, L) 
/\eps)+\tau(\eps,q,d)$, where the lower bound holds because two vertices at 
distance $d$ cannot be in the same part of any level smaller than $\log(d/2)$ 
due to the diameter bound of each part given by \cref{lem:talwar-decomp}. Using 
this property, it can be shown that $c$ and its closest center in the optimum
solution are also cut at a level in that range. As there are only $O(\log 
(1/\eps))+\tau(\eps,q,d)$ levels in this range, by the previous argument, the 
number of portals is now independent of $n$.

\subparagraph*{On highway dimension.} \label{par:tech-hd}
The above arguments for doubling metrics hold thanks 
to~\cref{lem:talwar-decomp}. In this work, we show how to construct a similar 
decomposition for low highway dimension: 

\begin{restatable}{lemma}{lemdecomp}
\label{lem:decomp}
Given a shortest-path metric $(V,\dist)$ of a graph with highway 
dimension~$\hw$, a subset $W\subseteq V$, and $\rho > 0$, there exists a 
polynomial-time computable randomized hierarchical 
decomposition~$\mc{D}=\{\mc{B}_0,\ldots,\mc{B}_{\lambda(W)}\}$ of $W$ such 
that the diameter of each part $B\in\mc{B}_i$ is at most $2^{i+5}$, and:
  \begin{enumerate}
  \item \textbf{Scaling probability}: for any 
$v\in V$, radius $r$, and level $i$, we have\\
$\Pr[\mc{D}\text{ cuts } \beta_v(r)\text{ at level } i]\leq \sigma \cdot r/2^i$, 
where $\sigma=(\hw\log(1/\rho))^{O(1)}$.
  \item \textbf{Interface}: for any $B\in\mc{B}_i$ %
on level $i\geq 1$ there exists an interface $I_B\subseteq V$, which is
    \begin{enumerate}
    \item \textbf{concise}: $|I_B|\leq (\hw/\rho)^{O(1)}$, 
and
    \item \textbf{precise}: for any $u,v\in B$ such that $u$ and $v$ 
are cut by $\mc{D}$ at level~$i-1$, there exists $p\in I_B$ with 
$\dist(u,p) + \dist(p,v) \le \dist(u,v) + 68\cdot\rho2^i$.
    \end{enumerate}
  \end{enumerate}
\end{restatable}

Our construction relies on the \emph{town decomposition} from 
\citep{FeldmannFKP15}, which is a laminar family with the following properties 
(see \cref{sec:prelim} for formal definitions and lemmas). If $\mc{T}$ is a town 
decomposition of a metric $(V,\dist)$, then every $T\in\mc{T}$ is a subset of 
$V$ and is called a \emph{town}, every vertex $v\in V$ is contained in at least 
one town, and also the whole set $V$ is a town of $\mc{T}$. Similar to 
hierarchical decompositions, the laminar family $\mc{T}$ thus decomposes $V$. If 
the metric is given by a graph of highway dimension~$\hw$, for a given $\rho > 
0$ every town $T\in\mc{T}$ has a set $X_T$ of \emph{hubs} with doubling 
dimension $O(\log(\hw \log(1/\rho)))$, such that for any two vertices $u$ and 
$v$ in different child towns of $T$, there is a hub $x\in X_T$ such that 
$\dist(u,x)+\dist(x,v)\leq (1+2\rho)\cdot \dist(u,v)$.

This hub set $X_T$ is similar to the portal set of \cref{lem:talwar-decomp}, but 
has some fundamental differences: the first one is that the town decomposition 
is deterministic, and so it may happen that a client and its facility are cut at 
a very high level --- something that happens only with tiny probability in the 
doubling setting thanks to the scaling probability. Another main difference is 
that the size of $X_T$ might be unbounded. As a consequence, it cannot be 
directly used as a portal set in a dynamic program. To deal with this, we 
combine the town decomposition with a hierarchical decomposition of each set 
$X_T$ according to \cref{lem:talwar-decomp}, to build an \emph{interface} as 
stated in \cref{lem:decomp}. 

A further notable difference to portals is that the preciseness property of the 
resulting interface is weaker. In particular, while there is a portal close 
to each vertex of a part, the hubs (and consequently the interface points) can 
be far from some vertices as long as they lie close to the shortest path to 
other vertices. This means that no analogue to near-optimal portal-respecting 
paths exist (see \cref{ap:hwdef}). Instead, when connecting a client $c$ with a 
facility~$f$ we need to use the interface point of~$I_B$, provided by the 
preciseness property of \cref{lem:decomp}, that lies close to the shortest path 
between $c$ and~$f$ for the lowest level part $B$ containing both $c$ and $f$. 
This shifts the perspective from externally connecting vertices of a part to 
vertices outside a part, as done for portals, to internally connecting vertices 
of parts, as done here.

As a consequence, we develop a dynamic program, which follows more or less 
standard techniques as for instance given in 
\cite{AroraRagRao98,kolliopoulos2007nearly}, but needs to handle the weaker 
preciseness property of the interface. The main idea is to guess the distances 
from interface points to facilities while recursing on the decomposition 
$\mc{D}$ of \cref{lem:decomp}. The runtime of this algorithm is thus 
exponential in the number of interface points. Thanks to the techniques 
developed by \citet{DBLP:conf/focs/SaulpicCF19} as described above, we can 
assume that this number is constant for clustering problems. However, due to the 
shifted perspective towards internally connecting vertices of parts, the runtime 
of the dynamic program also is exponential in the \emph{total} number of levels. 
It can be shown though that it suffices to compute a solution on a carefully 
chosen subset $W$ of the metric for which only a logarithmic number of levels of 
the decomposition need to be considered. Thus the overall runtime is bounded by 
some constant raised to a logarithm, which is polynomial.

\subsection{Outline}

After defining the concepts we use and stating various structural lemmas in 
\cref{sec:prelim}, we show how to incorporate our decomposition into the 
framework of \citep{DBLP:conf/focs/SaulpicCF19}. The proof of \cref{lem:decomp} 
is then presented in \cref{sec:decomp}. The formal algorithm can be found in 
\cref{sec:alg}. We conclude with the hardness proof of \cref{thm:hardness} in 
\cref{sec:hard}.

\section{Preliminaries}\label{sec:prelim}

\subparagraph*{On doubling metrics.}
The \emph{doubling dimension} of a metric is the smallest integer~$d$ such that 
for any $r>0$ and $v\in V$, the ball $\beta_v(2r)$ of radius $2r$ around $v$ can 
be covered by at most $2^d$  balls of half the radius $r$. A doubling metric is 
a metric space where the doubling dimension is constant. In those spaces one 
can show the existence of small \emph{nets}:

\begin{definition}\label{def:net}
A \emph{$\delta$-net} of a metric $(V,\dist)$ is a subset of nodes $N\subseteq 
V$ with the property that every node in $V$ is at distance at most $\delta$ from 
a net point of $N$, and each pair of net points of $N$ are at distance more 
than~$\delta$.
\end{definition}

Note that a simple greedy algorithm can compute a $\delta$-net for any given 
metric in polynomial time. In low doubling metrics these nets have the 
following useful properties, as shown by \citet{gupta2003bounded}.

\begin{lemma}[\cite{gupta2003bounded}]\label{prop:doub:net}
Let $(V,\dist)$ be a metric space with doubling dimension $d$. If $N\subseteq 
V$ is a $\delta$-net of diameter at most $D$, then $|N| \leq 2^{d \cdot \lceil 
\log_2 (D/\delta)\rceil}$. Moreover, any subset $W\subseteq V$ has doubling 
dimension at most $2d$.
\end{lemma}

\subparagraph*{On highway dimension.}
For simplicity we will set $c=8$ in \cref{dfn:hd} throughout this paper, even if 
all claimed results are also true for other values of $c$. When we refer to a 
metric as having highway dimension $\hw$, we mean that it is the shortest-path 
metric of a graph of highway dimension $\hw$. A \emph{laminar family} of $V$ is 
a set system with universe $V$ in which no two sets cross, i.e., any two sets 
are either disjoint or one set is contained in the other. This naturally gives 
rise to a rooted tree structure on the sets, and we thus refer to proper subsets 
of a set as its \emph{descendants} and to inclusion-wise maximal proper subsets 
as its \emph{children}.
The main result we will use about highway dimension is the existence of the 
following decomposition:

\begin{theorem}[\cite{FeldmannFKP15}]
\label{thm:apx-core-hubs}
Given a shortest-path metric $(V,\dist)$ of highway dimension~$\hw$, and $\rho > 
0$, there exists a polynomial-time computable deterministic laminar 
family~$\mc{T}$ of $V$, called the \emph{town decomposition}, where every set 
$T\in\mc{T}$ is called a \emph{town}. For every vertex $v\in V$ there is a 
singleton town~$\{v\}\in\mc{T}$, and also $V\in\mc{T}$. Every town $T$ has a set 
of \emph{hubs}\footnote{called \emph{approximate core hubs} in 
\cite{FeldmannFKP15}.} $X_T\subseteq T$ with the following properties:
\begin{enumerate}[a.]
\item \textbf{doubling}: the doubling dimension of $X_T$ is $d = 
O(\log(\hw\log(1/\rho)))$, and
\item \textbf{precise}: for any two vertices $u$ and $v$ in different child 
towns of $T$, there is a vertex $x\in X_T$ such that $\dist(u,x)+\dist(x,v)\leq 
(1+2\rho)\cdot \dist(u,v)$.
\end{enumerate}
\end{theorem}

The town decomposition behaves differently from those in 
\cref{lem:talwar-decomp,lem:decomp} in several ways. The main properties we will 
need here are given by the following lemma. Given a sequence of towns 
$T_0,\ldots, T_g$ of the towns decomposition such that $T_\ell$ is a child town 
of $T_{\ell-1}$ for each $\ell\in\{1,\ldots,g\}$, we call $T_g$ a 
\emph{$g$\textsuperscript{th}-generation descendant} of $T_0$. In particular, a 
child town is a $1$\textsuperscript{st}-generation descendant. The given 
property on these descendants is implicit in \cite{FeldmannFKP15} and we give a 
proof outline in \cref{ap:proof}.\footnote{We note that in the conference 
version of this paper~\cite{DBLP:conf/esa/FeldmannS20} it was erroneously 
claimed that for any child town $T'$ of a town $T$ we have 
$\diam(T')<\diam(T)/2$, while \cref{lem:townproperties} only gives 
$\diam(T')<\diam(T)$ in this case.}

\begin{restatable}[\cite{FeldmannFKP15}]{lemma}{lemtownproperties}
\label{lem:townproperties}
For any $T\in\mc{T}$ we have $\diam(T)<\dist(T,V\setminus T)$. Furthermore, for 
any positive integer $g$, if~$T'$ is a $g$\textsuperscript{th}-generation 
descendant town of $T$, then $\diam(T')<\diam(T)/2^{g-1}$.
\end{restatable}

\subparagraph*{On how to incorporate our decomposition into the framework of  
\citep{DBLP:conf/focs/SaulpicCF19}.}
Assume we are given an instance $\mc{I}$ of \cluster or \fl on some metric 
$(V,\dist)$, together with a hierarchical decomposition $\mc{D}$ of the metric 
with the properties listed in \cref{lem:decomp}. We start by defining the 
\emph{badly cut} clients. In the following, we fix an optimum solution $\opt$ 
and an approximate solution~$L$, and we define $\offset = 
\log_2(\sigma(q+1)^q/\eps^{q+1})$. Note that we will later use a constant 
approximation for $L$ (cf.~\cref{lem:coarse}), but for now the approximation 
ratio does not matter.

\begin{definition}[badly cut 
\cite{DBLP:conf/focs/SaulpicCF19}]\label{def:badlycut}
Let $(V,\dist)$ be a metric of an instance $\mc{I}$ of \cluster or \fl, $\mc{D}$ 
be a hierarchical decomposition of the metric with scaling probability 
factor~$\sigma$, and~$\eps > 0$. If $L_v$ is the distance from~$v$ to the 
closest facility of an approximate solution~$L$ to $\mc{I}$, then a client $c$ 
is \emph{badly cut w.r.t.~$\mc{D}$} if the ball $\beta_c(3L_c / \eps)$ is cut at 
some level $i$ greater than $\log_2(3L_c/ \eps) + \offset$.

Similarly, if $\opt_v$ is the distance from $v$ to the closest facility of the 
optimum solution $\opt$ of $\mc{I}$, then a facility $f\in L$ is \emph{badly 
cut w.r.t.~$\mc{D}$} if $\beta_f(3\opt_f)$ is cut at some level~$i$ greater 
than $\log_2 (3\opt_f) + \offset$.
\end{definition}

Given an instance $\mc{I}$ of \cluster or \fl and a decomposition $\mc{D}$ of 
the metric, a new instance $\mc{I}_\mc{D}$ is computed to get rid of badly cut 
clients. The instance $\mc{I}_\mc{D}$ is built from~$\mc{I}$ by moving clients 
that are badly cut w.r.t.~$\mc{D}$ to their closest facility in $L$. More 
concretely, let $\chi_\mc{I}$ and $\chi_{\mc{I}_\mc{D}}$ be the demand functions 
of $\mc{I}$ and~$\mc{I}_\mc{D}$, respectively. Initially we let $\mc{I}_\mc{D}$ 
be a copy of $\mc{I}$, so that in particular $\chi_{\mc{I}_\mc{D}}=\chi_\mc{I}$. 
Then, for each client $c$ of $\mc{I}$ that is badly cut in $L$ w.r.t.~$\mc{D}$, 
if $L(c)$ denotes the closest facility of $L$ to~$c$, in $\mc{I}_\mc{D}$ we set 
$\chi_{\mc{I}_\mc{D}}(c)=0$ and increase $\chi_{\mc{I}_\mc{D}}(L(c))$ by the 
value of $\chi_\mc{I}(c)$ in $\mc{I}$. For any client~$c$ of $\mc{I}_\mc{D}$ we 
denote by $\tilde{c}$ the original position of this client in $\mc{I}$, i.e., if 
$\tilde{c}$ is a badly cut client of $\mc{I}$ then $c=L(\tilde{c})$ and 
otherwise $c=\tilde{c}$.
The set $F$ of potential centers in unchanged, and thus any solution of $\mc{I}$ 
is a solution of $\mc{I}_\mc{D}$, and vice versa. Note that $\mc{I}_\mc{D}$ does 
not contain any badly cut client w.r.t.~$\mc{D}$, and that the definition of 
$\mc{I}_\mc{D}$ depends on the randomness of $\mc{D}$.

To describe the properties we obtain for the new instance, given a solution $S$ 
to any instance $\mc{I}_0$ of \cluster or \fl, we define 
$\cost_{\mc{I}_0}(S)=\sum_{v\in V}\chi_{\mc{I}_0}(v)\cdot\dist(v,S)^q$ to be the 
cost incurred by only the distances to the facilities. 
Note that for \cluster 
this coincides with the objective function, while for \fl we need to also add 
the facility opening costs to $\cost_{\mc{I}_0}(S)$ to obtain the objective function. 
Given some $\eps>0$ and the computed instance $\mc{I}_\mc{D}$ from 
$\mc{I}$, we 
define
\[
\nu_{\mc{I}_\mc{D}} = \max_{\text{solution }S} \big\{\cost_\mc{I}(S) -  
(1+2\eps) \cost_{\mc{I}_\mc{D}} (S)\ ,\ (1-2\eps) \cost_{\mc{I}_\mc{D}} (S) - 
\cost_\mc{I}(S)\big\}.
\]
If $B_\mc{D}$ denotes the set of badly cut facilities (w.r.t $\mc{D}$) of the 
solution $L$ to $\mc{I}$ from which instance~$\mc{I}_\mc{D}$ is constructed, we 
say that $\mc{I}_\mc{D}$ has \emph{small distortion w.r.t.~$\mc{I}$} if 
$\nu_{\mc{I}_\mc{D}} \leq \eps \cost_\mc{I}(L)$, and there exists a 
\emph{witness solution} $\hat S\subseteq F$ that contains $B_\mc{D}$ and for 
which
$
\cost_{\mc{I}_\mc{D}}(\hat S) \leq (1+O(\eps))\cost_{\mc{I}}(\opt) + O(\eps) 
\cost_{\mc{I}}(L).
$
Moreover, in the case of \fl, $\hat{S}=\opt\cup B_\mc{D}$ and $\sum_{f \in 
B_{\mc{D}}} w_f \leq \eps\cdot \sum_{f \in L} w_f$. 

Based on these definitions, we now state the main tool we use from 
\cite{DBLP:conf/focs/SaulpicCF19}, and which exploits the scaling probability 
of our decomposition in \cref{lem:decomp} to obtain the required 
structure.

\begin{lemma}[\cite{DBLP:conf/focs/SaulpicCF19}]\label{lem:struct}
Let $(V,\dist)$ be a metric, and $\mc{D}$ be a randomized hierarchical 
decomposition of $(V,\dist)$ with scaling probability factor $\sigma$. Let 
$\mc{I}$ be an instance of \cluster or \fl on $(V,\dist)$, with optimum solution 
$\opt$ and approximate solution $L$.
For any (sufficiently small) $\eps>0$, with probability at least $1-\eps$ 
(over~$\mc{D}$), the instance~$\mc{I}_\mc{D}$ constructed from $\mc{I}$ and $L$ 
as described above has small distortion with a witness solution~$\hat S$. 
Furthermore, for every client $\tilde{c}$ of $\mc{I}$ and corresponding client 
$c$ of $\mc{I}_\mc{D}$, client $c$ is cut by~$\mc{D}$ from its closest facility 
in $\hat S$ at level at most 
$\log_2(3L_{\tilde{c}}/\eps+4\opt_{\tilde{c}})+\offset$.
\end{lemma}

As a consequence of \cref{lem:struct}, a dynamic program can compute a solution 
recursively on the parts of $\mc{D}$ in polynomial time, as sketched in 
\cref{par:tech-hd} and detailed in \cref{sec:alg}.

\section{Decomposing the graph}
\label{sec:decomp}

This section is dedicated to the proof of \cref{lem:decomp}, which we restate 
here for convenience.

\lemdecomp*

To prove this, we first fix a town decomposition $\mc{T}$ of the input graph, as 
given by \cref{thm:apx-core-hubs} assuming (for technical reasons) that 
$\rho\leq 1/2$.
The general idea to construct a hierarchical decomposition~$\mc{D}$ is as 
follows. For doubling metrics, to decompose a part at level~$i$, it is enough to 
pick a random diameter $\delta \in [2^{i-2}, 2^{i-1})$ and divide the part into 
child parts of diameter~$\delta$. This is not doable in the highway dimension 
setting: if one wishes to decompose a town~$T$, it cannot divide any of the 
child towns, since it is not possible to use the hubs $X_T$ of $T$ to 
approximate paths inside one of the child towns. The high-level picture of our 
decomposition is therefore as follow. To decompose a town at level $i$, we group 
the ``small'' child towns randomly (as in the doubling decomposition), and put 
every other child town in its own subpart. As we will see, this turns out to be 
enough.

In order to decompose a town $T$, we need the following definitions.
For each child town~$T'$ of $T$ we identify the \emph{connecting hub} $x\in 
X_T$, which is some fixed closest hub of $X_T$ to $T'$, breaking ties 
arbitrarily. Moreover, given a hierarchical decomposition 
$\mc{X}_T=\{\mc{U}_0,\ldots,\mc{U}_{\lambda(X_T)}\}$ of~$X_T$, we define for 
every $i$ the \emph{connecting $i$-cluster} of a child town $T'$ of~$T$ to be 
the set~$U\in\mc{U}_\ell$ on level $\ell=\min\{i,\lambda(X_T)\}$ containing the 
connecting hub of~$T'$.
For a given subset $W\subseteq V$ we then follow the steps below, after choosing 
$\mu$ from the interval~$(0,1]$ uniformly at random (cf.~\cref{fig:decomp}):

\begin{enumerate}
 \item For each town $T\in\mc{T}$, we apply \cref{lem:talwar-decomp} to find a 
randomized hierarchical decomposition 
$\mc{X}_T=\{\mc{U}_0,\ldots,\mc{U}_{\lambda(X_T)}\}$ of the  hubs~$X_T$ of $T$.

\item In this step we fix a town $T\in\mc{T}$. Using $\mc{X}_T$, we define a 
randomized partial decomposition of~$T\cap W$ as follows.
For any~$i$ and $U\in\mc{U}_{\min\{i,\lambda(X_T)\}}$, let the set 
$A^U_i\subseteq T\cap W$ be the union of all sets $T'\cap W$ where $T'$ is a 
child town of $T$ with the following two properties:
\begin{enumerate}
 \item $U$ is the connecting $i$-cluster of $T'$, and
 \item $\dist(T',V\setminus T')\leq\mu 2^i$.
\end{enumerate}
In particular, $A^U_i$ contains towns somewhat close to $U$ due to~(a) and with 
small diameter due to~(b) and \cref{lem:townproperties}. We let $\mc{A}^T_i$ be 
the set containing every non-empty $A^U_i$.

\item Now, the hierarchical decomposition 
$\mc{D}=\{\mc{B}_0,\ldots,\mc{B}_{\lambda(W)}\}$ of $W$ can be constructed inductively as follows. 
At the highest level~$\lambda(W)$ of $\mc{D}$, $W$ is partitioned into a single 
set: $\mc{B}_{\lambda(W)}=\{W\}$. To decompose a part $B\in\mc{B}_{i+1}$ at 
level $i+1$, we do the following. Let $T\in\mc{T}$ be the inclusion-wise minimal 
town for which $B\subseteq T$. The ``small'' child towns of $T$ lying 
inside~$B$ are grouped according to step 2 (note that $\dist(T',V\setminus T')$ 
also bounds the diameter of $T'$ by \cref{lem:townproperties}), and the other 
ones form individual subparts (note though that these remaining child towns may 
also have small diameter and thus it would be misleading to call them ``big''). 
More formally, the set $\mc{B}_i$  contains every part $A\in\mc{A}^T_i$ for 
which $A\subseteq B$, and also every set $T'\cap W$, where $T'$ is a child town 
of $T$ for which $T'\cap W\subseteq B$ and $T'\cap W$ was not covered by the 
previously added parts of $\mc{A}^T_i$, i.e.,~$T'\cap W\cap A=\emptyset$ for 
every $A\in\mc{A}^T_i$. 

\end{enumerate}

To prove that the constructed decomposition $\mc{D}$ has the desired properties 
-- i.e. that it is indeed a hierarchical decomposition with parts of bounded 
diameter and small scaling probability --
we begin with some auxiliary lemmas, of which the first one 
bounds the distance of a town to its connecting hub.

\begin{figure}
\centering
\includegraphics[width=6cm]{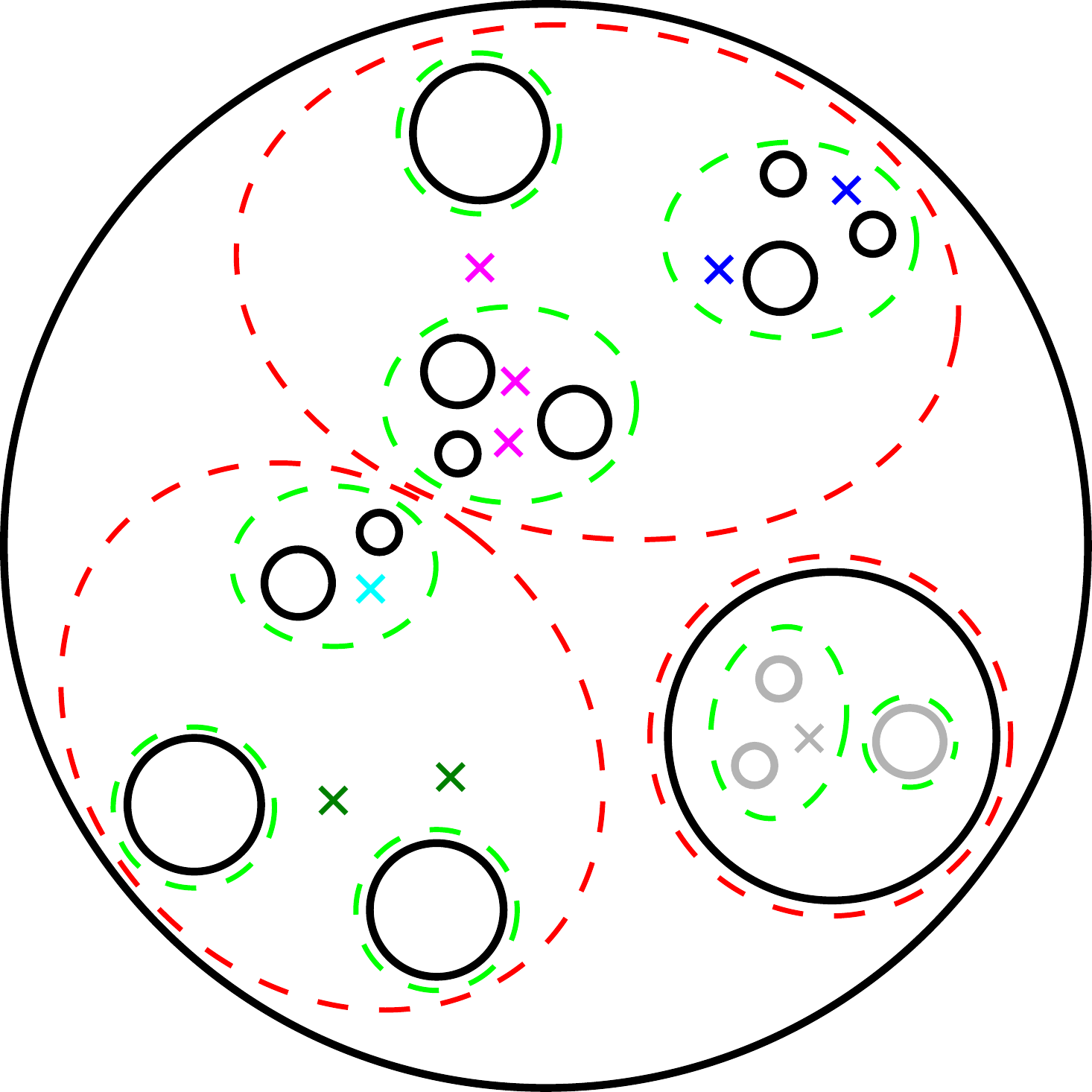}
\caption{A town and its child towns (black circles). The hubs (crosses) are 
decomposed by $\mc{X}_T$, indicated by different colours (note that in 
reality the hubs are contained in child towns, but are depicted separately here 
for clarity). Parts $B\in\mc{B}_{i+1}$ (red dashed) are decomposed into parts on 
level~$i$ (green dashed). Parts of~$\mc{B}_i$ can lie in different towns (e.g., 
the child town of $T$ containing grandchild towns in grey).}
\label{fig:decomp}
\end{figure}

\begin{lemma}%
\label{lem:town-dists}
If $T'$ is a child town of $T$ with connecting hub $x\in X_T$, then we have
$\dist(x,T')\leq (1+2\rho)\dist(T',V\setminus T')$.
\end{lemma}
\begin{proof}
Let $T''$ be the closest sibling town to $T'$, and let $u\in T'$ and $v\in T''$ 
be the vertices defining the distance from $T'$ to $T''$, i.e., 
$\dist(u,v)=\dist(T',T'')=\dist(T',V\setminus T')$. By 
\cref{thm:apx-core-hubs}, there is a hub $y\in X_T$ for which 
$\dist(u,y)+\dist(y,v)\leq(1+2\rho)\cdot\dist(u,v)=(1+2\rho)\cdot\dist(T',
V\setminus T')$. This implies $\dist(y,T')\leq \dist(u,y)\leq 
(1+2\rho)\cdot\dist(T',V\setminus T')$. Since the connecting hub $x$ of $T'$ is 
at least as close to $T'$ as $y$, the claim follows.
\end{proof}

Based on the above lemma, we next prove the key property that the diameter of 
any part of $\mc{B}_i\in\mc{D}$ is bounded.

\begin{lemma}%
\label{lem:diam}
If $\rho\leq 1/2$, then the diameter of any part of $\mc{B}_i\in\mc{D}$ is less 
than $2^{i+5}$.
\end{lemma}
\begin{proof}
On the highest level $\lambda(W)$ of $\mc{D}$ the only part of 
$\mc{B}_{\lambda(W)}$ is $W$ itself. As $\lambda(W)=\lceil\log_2\diam(W)\rceil$ 
we get $\diam(W)\leq 2^{\lambda(W)+1}$, as required.

For any level $i<\lambda(W)$, a set in $\mc{B}_i$ is either equal to a set 
$A\in\mc{A}^T_i$ for some town~$T\in\mc{T}$ or it is equal to some set $T'\cap 
W$ for a child town $T'$ of $T$. In the former case, the set~$A$ is a 
set~$A^U_i$ for some cluster $U\in\mc{U}_\ell$ where 
$\ell=\min\{i,\lambda(X_T)\}$ and $\mc{U}_\ell\in\mc{X}_T$. The set~$A^U_i$ 
contains the union of sets $T'\cap W$ for child towns $T'$ of~$T$, for which 
their connecting hubs lie in~$U$ and $\dist(T',V\setminus T')\leq\mu 2^i\leq 
2^i$, as~$\mu\leq 1$. Thus from \cref{lem:town-dists} we get $\dist(U,T')\leq 
(1+2\rho)2^i$, and by \cref{lem:townproperties} we have 
$\diam(T')<\dist(T',V\setminus T')\leq 2^i$. The cluster~$U$  has diameter 
at most $2^{i+1}$ by \cref{lem:talwar-decomp}, since it is part of the 
hierarchical decomposition~$\mc{X}_T$ and lies on level $\ell\leq i$. Let $u$ 
and $v$ be the vertices of $A^U_i$ defining the diameter of $A^U_i$, i.e., 
$\dist(u,v)=\diam(A^U_i)$. We may reach $v$ from $u$ by first crossing the child 
town $T'$ that $u$ lies in, then passing over to $U$, then crossing~$U$, after 
which we pass over to the child town $T''$ containing~$v$, and finally crossing 
this child town as well to reach~$v$. Hence, assuming that $\rho\leq 1/2$ the 
diameter of $A^U_i$ is bounded by
\[
\dist(u,v)\leq
\diam(T')+\dist(U,T')+\diam(U)+\dist(U,T'')+\diam(T'')\\
< 2\cdot2^i +2\cdot(1+2\rho)2^i+2^{i+1}=(6+4\rho)2^{i} \leq 2^{i+3}.
\]

Now consider the other case, when a set $B\in\mc{B}_i$ on level $i<\lambda(W)$ 
is equal to some set $T'\cap W$ for a child town $T'$ of a town $T$. For such a 
child town $T'$ there is no enforced upper bound on the distance to other child 
towns as before, and thus it is necessary to be more careful to bound the 
diameter of the part. Starting with $B=B_i$, let $B_i\subseteq 
B_{i+1}\subseteq\ldots\subseteq B_j$ be the longest chain of parts of increasing 
levels that are of the same type as $B$. More concretely, for every 
$\ell\in\{i,i+1,\ldots,j\}$ we have $B_\ell\in\mc{B}_\ell$ and $B_\ell$ is equal 
to some set $T'_\ell\cap W$ for a child town~$T'_\ell$ of the inclusion-wise 
minimal town $T_\ell$ containing $B_{\ell+1}$. Note that in particular 
$j<\lambda(W)$. As we chose the longest 
such chain, on the next level $j+1$ there is no such set containing $B_j$, which 
means that the set $B_{j+1}\in\mc{B}_{j+1}$ for which $B_j\subseteq B_{j+1}$ is 
either equal to a set $A\in\mc{A}^{T_{j+1}}_{j+1}$ for some town $T_{j+1}$, or 
$j+1=\lambda(W)$. In either case, from above we get $\diam(B_{j+1})\leq 
2^{j+4}$.

Note that for any $\ell\in\{i,i+1,\ldots,j-1\}$, since 
$B_{\ell+1}=T'_{\ell+1}\cap W$ implies $B_{\ell+1}\subseteq T'_{\ell+1}$, while 
$T_\ell$ is the inclusion-wise minimal town containing $B_{\ell+1}$, we have 
$T_\ell\subseteq T'_{\ell+1}$. Now, as $T'_\ell$ is a child town of~$T_\ell$, we 
get that $T'_\ell$ is a descendant of $T'_{\ell+1}$. This means that $T'_i$ is a 
$g$\textsuperscript{th}-generation descendant of $T'_j$ for some $g\geq j-i$, 
and from \cref{lem:townproperties} we get 
$\diam(T'_i)\leq\diam(T'_j)/2^{j-i-1}$. As $B=B_i\subseteq T'_i$ we have 
$\diam(B)\leq\diam(T'_i)$. Since $T_j$ is the inclusion-wise minimal town 
containing $B_{j+1}$, the latter set contains vertices of at least two child 
towns of $T_j$. One of these child towns is $T'_j$, since $B_j=T'_j\cap W$ and 
$B_j\subseteq B_{j+1}$ by construction of the decomposition. In particular, 
$B_{j+1}$ both contains vertices inside and outside of $T'_j$, and so 
$\dist(T'_j,V\setminus T'_j)\leq\diam(B_{j+1})$. By \cref{lem:townproperties} we 
know that $\diam(T'_j)<\dist(T'_j,V\setminus T'_j)$, and putting all these 
inequalities together we obtain
\[
\diam(B)\leq\diam(T'_i)\leq\diam(T'_j)/2^{j-i-1}<
\dist(T'_j,V\setminus T'_j)/2^{j-i-1}\\
\leq \diam(B_{j+1})/2^{j-i-1}\leq 2^{j+4}/2^{j-i-1}= 2^{i+5}.
\qedhere
\]
\end{proof}

Using \cref{lem:diam} it is not hard to prove the correctness of $\mc{D}$, which 
we turn to next.

\begin{lemma}%
\label{lem:correctness}
The tuple $\mc{D}=\{\mc{B}_0,\ldots,\mc{B}_{\lambda(W)}\}$ is a hierarchical 
decomposition of $W$.
\end{lemma}
\begin{proof}%
We first prove that for a part $B \in \mc{B}_i$ included in town $T$, part $B$ 
can be partitioned into unions of sets $T'\cap W$ for child towns $T'$ of $T$. 
Indeed, either $B = T\cap W$, and properties of the town decomposition ensure 
that $B$ can be partitioned in this way, or~$B \in \mc{A}^T_i$. By construction 
of~$\mc{A}^T_i$, in the latter case part $B$ is also the union of sets $T'\cap 
W$ for child towns $T'$ of $T$.

Now, step (3) of the construction decomposes $B$ into groups of child towns 
restricted to~$W$, and so $B$ is partitioned by $\mc{B}_{i-1}$. Moreover, since 
$\mc{B}_{\lambda(W)}=\{W\}$, by induction each $\mc{B}_i$ is a partition of~$W$. 
That concludes the proof.
\end{proof}

We now turn to proving the properties of \cref{lem:decomp}, starting with the 
scaling probability.%

\begin{lemma}
The decomposition $\mc{D}$ has scaling probability factor 
$\sigma=(\hw\log(1/\rho))^{O(1)}$.
\end{lemma}
\begin{proof}
To prove the claim, we need to prove that for any $v\in W$, radius~$r$, and 
level~$i$, the probability that $\mc{D}$ cuts the ball $\beta_v(r)$ at level $i$ 
is at most $(\hw\log(1/\rho))^{O(1)}\cdot r/2^i$. If $\mc{D}$ cuts $\beta_v(r)$ 
at level $i$, it means that $\beta_v(r)$ is fully contained in a part at level 
$i+1$. Let $T\in\mc{T}$ be the inclusion-wise minimal town containing that part.
There are two cases to consider: either $\beta_v(r)$ is cut by ``small'' parts, 
i.e. there exist two distinct parts $A_1,A_2\in\mc{A}^T_i$
such that $v\in A_1$ and $u\in A_2$ for some $u\in W\cap\beta_v(r)$, or 
not.

We start with the latter case, when $\beta_v(r)$ is not cut by small parts. If 
$\mc{D}$ cuts the ball at level~$i$, there are distinct parts 
$B,B'\in\mc{B}_i$ such that $v\in B$ and $u\in B'$ for 
some~$u\in W\cap\beta_v(r)$.  Assume 
w.l.o.g.~that $B\notin\mc{A}^T_i$ (which is possible to assume since $\beta_v(r)$ is not cut by small parts). 
By construction of the decomposition, there must be a child town $T'$ of~$T$, 
for which $B= T'\cap W$ and $\dist(T',V\setminus T')>\mu 2^i$. Note that  
$r \geq \dist(v,u)\geq\dist(T',B')\geq\dist(T',V\setminus T') > \mu 2^i$, and 
hence $\mu < r/2^i$. The decomposition~$\mc{D}$ can therefore only cut 
$\beta_v(r)$ on level~$i$ if~$\mu<r/2^i$. Since $\mu$ is chosen uniformly at 
random from the interval $(0,1]$, the probability is less than $r/2^{i}$.

We now turn to the other case when $\beta_v(r)$ is cut by two small parts 
$A_1,A_2\in\mc{A}^T_i$. The town~$T$ must have two child towns $T_1$ and~$T_2$ 
for which $v\in T_1\cap W\subseteq A_1$ and $u\in T_2\cap W\subseteq A_2$. Let 
$x_1$ and~$x_2$ be the connecting hubs of $T_1$ and $T_2$, respectively. The 
decomposition $\mc{D}$ cuts $v$ and~$u$ on level~$i$ if and only if $\mc{X}_T$ 
cuts $x_1$ and $x_2$ on level $\ell=\min\{i,\lambda(X_T)\}$. Indeed, let $U_1$ 
and $U_2$ be the connecting $i$-clusters of $T_1$ and~$T_2$, respectively, so 
that $A_1=A^i_{U_1}$ and $A_2=A^i_{U_2}$ with $x_1 \in U_1$ and $x_2 \in U_2$.  
Thus $\mc{D}$ cuts $v$ and $u$ on level~$i$ if and only if $U_1\neq U_2$, i.e., 
if and only if $\mc{X}_T$ cuts $x_1$ and $x_2$ on level 
$\ell=\min\{i,\lambda(X_T)\}$.

To compute the probability that $x_1$ and $x_2$ are cut, it is necessary to bound the distance between 
them. As $v\in T_1$ and $u\in T_2$ while $u\in\beta_v(r)$, for each 
$j\in\{1,2\}$ we have $\dist(T_j,V\setminus T_j)\leq\dist(T_1,T_2)\leq r$. By 
\cref{lem:town-dists} the distance between $T_j$ and its connecting hub $x_j\in 
X_T$ is thus at most $(1+2\rho)r$. Also, by \cref{lem:townproperties} we have 
$\diam(T_j)<\dist(T_j,V\setminus T_j)\leq r$, and we get
\[
\dist(x_1,x_2)\leq 
\dist(x_1,T_1)+\diam(T_1)+\dist(T_1,T_2)+\diam(T_2)+\dist(T_2,x_2)<\\
2(1+2\rho)r + 2r + r = (5+4\rho)r.
\]

We can reformulate the above as follows: if $\mc{D}$ cuts the ball $\beta_v(r)$ 
at level $i$, and $\beta_v(r)$ is cut by some ``small'' parts $A_1$ and $A_2$, 
then $\mc{X}_T$ cuts the ball $\beta_{x_1}((5+4\rho)r)$ on level $i$, 
where~$x_1$ is the hub defined for $v$ above. We know that the probability of 
the latter event is at most $2^{O(d)} (5+4\rho)r/2^i$ by 
\cref{lem:talwar-decomp}, where  $d = O(\log(\hw\log(1/\rho)))$ is the doubling 
dimension of $X_T$  by \cref{thm:apx-core-hubs}. Hence the probability that 
$\mc{D}$ cuts the ball $\beta_v(r)$ by some ``small'' parts at level~$i$ is 
at most $(\hw\log(1/\rho))^{O(1)}\cdot r/2^i$.

Taking a union bound over the two considered cases proves the claim.
\end{proof}

To prove the remaining property of \cref{lem:decomp} for $\mc{D}$, for each 
$B\in\mc{B}_i$ we need to choose an interface $I_B$ from the whole vertex set 
$V$. For this we use a carefully chosen net (see~\cref{def:net}) of the hubs of the 
inclusion-wise minimal town~$T$ containing $B$, as formalized in the following 
lemma.

\begin{lemma}\label{lem:interface}
Given $B\in\mc{B}_i$ for some $\mc{B}_i\in\mc{D}$ where $i\geq 1$, let 
$T\in\mc{T}$ be the inclusion-wise minimal town containing $B$. We define the 
interface $I_B$ to be a $\rho 2^i$-net of the set $Y_B=\{x\in 
X_T\mid\dist(x,B)\leq (1+2\rho)\diam(B)\}$. The interface $I_B$ has the 
conciseness and preciseness properties of \cref{lem:decomp} for $\rho\leq 1/2$.
\end{lemma}
\begin{proof}
We first prove that $I_B$ is precise. Consider two vertices $u,v\in B$ that are 
cut at level~\mbox{$i-1$} by~$\mc{D}$. This means there are two distinct parts 
$B',B''\in\mc{B}_{i-1}$ on this level such that~$v\in B'$ and~$u\in B''$. By 
definition, both $B'$ and $B''$ are unions of sets $T'\cap W$ where $T'$ is a 
child town of the inclusion-wise minimal town $T$ containing~$B$. Also $B'\cap 
B''=\emptyset$ by \cref{lem:correctness}. This means that $T$ has two child 
towns $T_1$ and $T_2$ for which $v\in T_1\cap W\subseteq B'$ and \mbox{$u\in 
T_2\cap W\subseteq B''$}. By \cref{thm:apx-core-hubs}, there is a hub $x\in X_T$ 
such that $\dist(u,x)+\dist(x,v)\leq (1+2\rho)\dist(u,v)$. In particular, 
$\dist(x,B)\leq\dist(u,x)\leq (1+2\rho)\dist(u,v)\leq (1+2\rho)\diam(B)$, as 
$u,v\in B$. This means that~$x\in Y_B$. Since $I_B$ is a $\rho 2^i$-net of 
$Y_B$, there is a node $p\in I_B$ for which $\dist(x,p)\leq\rho 2^i$. By 
\cref{lem:diam} we have $\dist(u,v)\leq\diam(B)\leq 2^{i+5}$ if $\rho\leq 1/2$, 
and so $I_B$ is precise:
\[
\dist(u,p) +\dist(p,v)\leq \dist(u,x) +2\cdot\dist(x,p)+\dist(x,v)\\
\leq (1+2\rho)\dist(u,v)+\rho 2^{i+1}\leq
\dist(u,v)+ 2\rho\cdot 2^{i+5}+\rho 2^{i+1}
\leq\dist(u,v)+68\cdot\rho2^i.
\]

To prove conciseness, recall that $\diam(B)\leq 2^{i+5}$ by \cref{lem:diam}, 
which means that $\diam(Y_B)\leq\diam(B)+ 2(1+2\rho)\diam(B)\leq 5\cdot 2^{i+5}$ 
for $\rho\leq 1/2$. Since $I_B$ is a $\rho 2^i$-net of $Y_B$, 
\cref{prop:doub:net} implies $|I_B|\leq 2^{d\cdot\lceil 
\log_2(160/\rho)\rceil}$, where $d$ is the doubling dimension of~$Y_B$. 
\cref{thm:apx-core-hubs} says that $X_T$ has doubling dimension 
$O(\log(\hw\log(1/\rho)))$, and as $Y_B\subseteq X_T$ the same asymptotic bound 
holds for the doubling dimension $d$ of $Y_B$ by \cref{prop:doub:net}. Therefore 
we get $|I_B|\leq 2^{O(\log(\hw\log(1/\rho))\cdot (\log(1/\rho)))}\leq 
(\hw/\rho)^{O(1)}$, which concludes the proof.
\end{proof}

\section{The algorithm}
\label{sec:alg}

Let $\mc{I}$ be an instance of the \cluster or \fl problem on a shortest-path 
metric $(V,\dist)$ of a graph $G$ with highway dimension $\hw$, and maximum 
demand $X=\max_{v\in V}\chi_{\mc{I}}(v)$. Given $\mc{I}$ the algorithm performs 
the following steps:

\begin{enumerate}
\item compute a town decomposition $\mc{T}$ of the metric together with the 
hub sets $X_T$ for each town $T\in\mc{T}$ as given by \cref{thm:apx-core-hubs}.
\item compute a hierarchical decomposition $\mc{D}$ according to 
\cref{lem:decomp}. Simultaneously $\mc{I}$ is reduced (see 
\cref{sec:approx-dist}) to a \emph{coarse} instance w.r.t.~$\mc{D}$, meaning 
that there is a subset~$W\subseteq V$ for which
\begin{itemize}
\item the clients and facilities of $\mc{I}$ are contained in $W$, i.e., 
$F\cup\{v\in V\mid\chi_\mc{I}(v)>0\}\subseteq W$, and
\item  every part of~$\mc{D}$ on level at most 
$\xi(W)=\lfloor\lambda(W)-2\log_2(nX/\eps)\rfloor$ has at most one 
facility, i.e., $|B\cap F|\leq 1$ for every $B\in\mc{B}_{\xi(W)}$.
\end{itemize}
\item compute the instance $\mc{I}_\mc{D}$ of small distortion as given by 
\cref{lem:struct}.
\item run a dynamic program on $\mc{I}_\mc{D}$ as given in \cref{sec:dp}, to 
compute an \emph{optimum rounded interface-respecting solution} (see 
\cref{sec:approx-dist} for a formal definition), and convert it to a solution 
for the input instance.
\end{enumerate}

In a nutshell, the coarseness of the instance guarantees that only a logarithmic 
number of levels need to be considered by the dynamic program. 
Reducing to a coarse instance in step~(2) loses a $(1+\eps)$-factor in the 
solution quality. The dynamic program is only able to compute highly structured 
solutions, which are captured by the notions of \emph{rounding} and 
\emph{interface-respecting}. Due to this, another $(1+\eps)$-factor in the 
solution quality is lost. In \cref{sec:approx-dist} we prove that the output of 
the dynamic program corresponds to a near-optimal solution of the input instance 
(proving \cref{thm:main-alg}), and we also detail step~(2) of the algorithm. 
Then in \cref{sec:dp} we describe the dynamic program.

\subsection{Approximating the distances}\label{sec:approx-dist}

One caveat of the dynamic program is that the runtime is only 
polynomial if the the recursion depth is logarithmic. However when computing our 
decomposition on the whole metric $(V,\dist)$, the number of levels is 
$\lambda(V)+1=\lceil\log_2\diam(V)\rceil+1$, which can be linear in the input 
size. For general metrics, standard preprocessing techniques can be used to 
reduce the number of levels to $O(\log(n/\eps))$ when aiming for a 
$(1+\eps)$-approximation. However, for graphs of bounded highway dimension these 
general techniques change the hub sets and we would have to be careful to 
maintain the properties we need in order to apply 
\cref{thm:apx-core-hubs}.\footnote{We note that in~\cite{FeldmannFKP15} these 
general techniques are indeed applied to low highway dimension graphs, but some 
details of the argument are left out. Instead of rectifying the technique 
in~\cite{FeldmannFKP15}, here we chose to go via the route of coarse instances.} 
Therefore we adapt the standard techniques to our setting via the notion of 
\emph{coarse} instances.

The following lemma shows that we can reduce any instance to a set 
of coarse ones, for which, as we will see, our dynamic program only needs to 
consider the highest $2\log_2(nX/\eps)$ levels.

\begin{lemma}\label{lem:coarse}
Let $\mc{I}$ be an instance of \cluster or \fl on a graph~$G$ of highway 
dimension $\hw$. There are polynomial-time computable instances 
$\mc{I}_1,\ldots,\mc{I}_b$ and respective hierarchical decompositions 
$\mc{D}_1,\ldots,\mc{D}_b$ with the properties given in \cref{lem:decomp} for 
any $\rho\leq 1/2$, such that for each $i\in\{1,\ldots,b\}$ the 
instance~$\mc{I}_i$ is also defined on $G$ and is coarse w.r.t.~$\mc{D}_i$. 
Furthermore, if an $\alpha$-approximation can be computed for each of the 
instances $\mc{I}_1,\ldots,\mc{I}_b$ in polynomial time, then for any $\eps>0$ 
a $(1+O(\eps))\alpha$-approximation can be computed for $\mc{I}$ in polynomial 
time.
\end{lemma}
\begin{proof}
Let us first describe the construction of the instances 
$\mc{I}_1,\ldots,\mc{I}_b$. We begin by computing a constant approximation $L$ 
to the given instance $\mc{I}$ of \cluster or \fl, using a 
$\gamma$-approximation algorithm as given in~\cite{abs-0809-2554} where 
$\gamma\in O(1)$. 
Let $\Lambda$ be the value of the objective function of the approximate solution 
$L$, i.e., $\Lambda=\cost_\mc{I}(L)$ if $\mc{I}$ is an instance of \cluster and 
$\Lambda=\cost_\mc{I}(L)+\sum_{f\in L}w_f$ in case of \fl. Let $\Gamma$ be the 
objective function value of an optimum solution $\opt$ to $\mc{I}$. For every 
client~$c$ of $\mc{I}$ (for which $\chi_\mc{I}(c)>0$), we have 
$\dist(c,\opt)\leq\Gamma^{1/q}\leq\Lambda^{1/q}$. Hence if we 
consider the subgraph of $G$ spanned by all edges of length at most 
$\Lambda^{1/q}$, then the closest facility of $\opt$ to $c$ lies in the same 
connected component of the subgraph as $c$. 

Ideally, we would want each of these components to define an instance 
$\mc{I}_i$. However, such a component might not have bounded highway dimension 
and we would thus not be able to compute a hierarchical decomposition using 
\cref{lem:decomp} for it. Instead we use the same input graph~$G=(V,E)$, but 
restrict the client and facility sets to a component. More formally, let 
$W_1,\ldots,W_b\subseteq V$ be the vertex sets of the connected components of 
the subgraph of $G$ spanned by all edges of length at most $\Lambda^{1/q}$. Note 
that $\diam(W_i)\leq n\Lambda^{1/q}$. For each $i\in\{1,\ldots,b\}$ we define an 
instance $\mc{I}_i$ on $G$ with $\chi_{\mc{I}_i}(v)=\chi_{\mc{I}}(v)$ for every 
$v\in W_i$ and $\chi_{\mc{I}_i}(v)=0$ otherwise. Initially, the facility set 
$F_i$ of $\mc{I}_i$ is $F\cap W_i$, where $F$ is the facility set of $\mc{I}$. 
We still need to coarsen this set $F_i$ though, which we do next.

At this point we compute a hierarchical decomposition $\mc{D}_i$ of $W_i$ for 
each $\mc{I}_i$ using $G$ according to \cref{lem:decomp}, i.e., the interface 
sets are from $V\supseteq W_i$. To make $\mc{I}_i$ coarse w.r.t.~$\mc{D}_i$, 
consider a part $B\in\mc{B}_{\xi(W_i)}$ of $\mc{D}_i$ containing facilities from 
$F_i$. In case of \fl, let $f\in F_i\cap B$ be a facility of minimum 
weight~$w_f$ among those in $F_i\cap B$, and in case of \cluster, fix an 
arbitrary $f\in F_i\cap B$. We call $f$ the \emph{representative facility of 
$\mc{I}_i$} for the facilities in $F_i\cap B$, and remove all facilities other 
than $f$ in $F_i\cap B$ from the set~$F_i$. We repeat this for every part of 
$\mc{B}_{\xi(W)}$. Note that $W_i$ contains all facilities and clients 
of~$\mc{I}_i$, i.e., $F_i\cup\{v\in V\mid\chi_{\mc{I}_i}(v)>0\}\subseteq W_i$, 
and thus~$\mc{I}_i$ is now a coarse instance w.r.t.~$\mc{D}_i$.

To prove the second part of the lemma, consider the optimum solution $\opt$ to 
$\mc{I}$. For every~$i$ we define a solution $S_i^*$ to $\mc{I}_i$, which for 
each facility in $\opt\cap W_i$ contains the representative facility 
of~$\mc{I}_i$. Since in case of \fl the representative facility is the one of 
minimum opening cost in the respective part in $\mc{B}_{\xi(W_i)}$ of $\mc{D}_i$ 
and the facility sets of different instances are disjoint, we 
have~$\sum_{i=1}^b\sum_{f\in S^*_i}w_f\leq\sum_{f\in\opt}w_f$. Also, 
$\sum_{i=1}^b|S^*_i|\leq|\opt|$, which means that if~$\mc{I}$ is an instance of 
\cluster then each $\mc{I}_i$ should be an instance of 
\pname{$k_i$-Clustering$^q$} for $k_i=|S^*_i|$, where however we do not know the 
value of $k_i$ a priori. We later show how to deal with this.

To bound the connection costs, first note that as 
$\lambda(W_i)=\lceil\log_2(\diam(W_i))\rceil$, $\diam(W_i)\leq n\Lambda^{1/q}$, 
and $1/q<1+1/q\leq 2$ for $q\geq 1$ we have
\[
\xi(W_i)\leq\lambda(W_i)-2\log_2(nX/\eps)<
\log_2(n\Lambda^{1/q})+1+
\log_2\left(\frac{\eps^{1+1/q}}{n^{1+1/q}X^{1/q}}\right)\\
=\log_2\left(\eps^{1+1/q}\left(\Lambda/(nX)\right)^{1/q}\right)+1.
\]

Now consider any client $c$ of $\mc{I}_i$ and its closest facility 
$\hat{f}\in\opt$ in the optimum solution to~$\mc{I}$, for which we know that 
$c,\hat{f}\in W_i$. Let $f^*\in S^*_i$ be the representative facility 
of~$\hat{f}$, which lies in the same part $B\in\mc{B}_{\xi(W_i)}$ as~$\hat{f}$. 
By \cref{lem:diam} the diameter of~$B$ is less than $2^{\xi(W_i)+5}$ (if 
$\rho\leq 1/2$). Hence we have
\[
\dist(c,f^*)\leq \dist(c,\hat{f})+\diam(B)<\dist(c,\hat{f})+ 
64\cdot\left(\eps^{1+1/q}\left(\Lambda/(nX)\right)^{1/q}\right).
\]
To bound $\dist(c,f^*)^q$ we need the following fact taken 
from~\cite{Cohen-AddadS17}.

\begin{prop}[\cite{Cohen-AddadS17}]\label{faketriangleineq}
Given $x,y,q \ge 0$, and $0<\eps<1/2$ we have\\
$(x+y)^q \leq (1+\eps)^q x^q + (1+1/\eps)^q y^q$.
\end{prop}

For constant $q\geq 1$ we have $(1+\eps)^q=1+O(\eps)$ and 
$(1+1/\eps)^q=O(1/\eps^q)$ as $\eps$ tends to zero. Thus the bound of
\cref{faketriangleineq} can be stated as $(x+y)^q \leq (1+O(\eps)) x^q + 
O(1/\eps^q) y^q$ if $q\geq 1$, and we get
\[
\dist(c,f^*)^q &< (1+O(\eps))\dist(c,\hat{f})^q+ 
O(1/\eps^q)\left(\eps^{1+1/q}\left(\Lambda/(nX)\right)^{1/q}\right)^q\\
&=
(1+O(\eps))\dist(c,\hat{f})^q+ O\left(\frac{\eps\Lambda}{nX}\right).
\]
Using the definition of $\chi_{\mc{I}_i}(v)$, in addition to 
$\dist(c,S^*_i)\leq\dist(c,f^*)$, $\dist(c,\opt)=\dist(c,\hat{f})$, and 
$\sum_{v\in V}\chi_{\mc{I}}(v)\leq nX$, we obtain 
\[
\sum_{i=1}^b\cost_{\mc{I}_i}(S^*_i)&= 
\sum_{i=1}^b\sum_{v\in W_i} \chi_{\mc{I}_i}(v)\cdot\dist(v,S^*_i)^q\\
&< \sum_{v\in V}\chi_\mc{I}(v) \left((1+O(\eps))\cdot\dist(v,\opt)^q+ 
O\left(\frac{\eps\Lambda}{nX}\right)\right)\\
&\leq (1+O(\eps))\cost_{\mc{I}}(\opt) +O(\eps\Lambda).
\]

For \fl, applying an $\alpha$-approximation algorithm to each 
instance~$\mc{I}_i$ gives respective solutions $S_i$ for which
\[
\sum_{i=1}^b \Big(\cost_{\mc{I}_i}(S_i) + \sum_{f\in S_i} w_f\Big) &\leq 
\sum_{i=1}^b \alpha\Big(\cost_{\mc{I}_i}(S^*_i) + \sum_{f\in S^*_i} w_f\Big) \\
&< \alpha\Big((1+O(\eps))\cost_{\mc{I}}(\opt) 
+O(\eps\Lambda)+\sum_{f\in\opt}w_f\Big)
\]
As $\Lambda$ is the objective function value of a $\gamma$-approximation to 
$\opt$ where $\gamma$ is constant, this means that by taking $\bigcup_{i=1}^b 
S_i$ as a solution to $\mc{I}$ we obtain a $(1+O(\eps))\alpha$-approximation as 
required.

For \cluster we need to do more work, since we do not know the number of 
facilities~$k_i$ to be opened in each instance $\mc{I}_i$. First, for every 
$k'\in\{0,\ldots,k\}$ we compute an $\alpha$\hy{}approximation~$S_i(k')$ to 
\pname{$k'$-Clustering$^q$} on each instance $\mc{I}_i$, i.e., $S_i(k')\subseteq 
F_i$ and $|S_i(k')|\leq k'$, and define $A_i(k')=\cost_{\mc{I}_i}(S_i(k'))$ to 
be its objective function value. Now let~$A_{\leq i}(k')$ be of the minimum 
value of $\sum_{j=1}^i A_j(k'_j)$ over all tuples $k'_1,\ldots,k'_i$ for which 
$\sum_{j=1}^i k'_j=k'$. To compute~$A_{\leq i}(k')$ in polynomial time, we use 
the following simple recursion. For $i=1$ we clearly have $A_{\leq 
1}(k')=A_1(k')$, and for $i>1$ we have $A_{\leq i}(k')=\min\{A_{\leq 
i-1}(k'-k'_i)+A_{i}(k'_i) \mid 0\leq k'_i\leq k'\}$. Note that it takes 
$O(bk^2)$ time to compute all values $A_{\leq i}(k')$. Finally, for the input 
instance $\mc{I}$ we output the union $\bigcup_{i=1}^b S_i(k'_i)$ of solutions 
that obtain the value~$A_{\leq b}(k)$. By definition of $A_{\leq b}(k)$ this is 
a feasible solution with $k$ facilities, and we have $A_{\leq 
b}(k)\leq\sum_{i=1}^b A_i(k_i)$ for the values $k_i=|S^*_i|$. Thus
\[
\sum_{i=1}^b \cost_{\mc{I}_i}(S_i(k'_i))=A_{\leq b}(k)\leq 
\sum_{i=1}^b \cost_{\mc{I}_i}(S_i(k_i))\leq
\sum_{i=1}^b \alpha\cost_{\mc{I}_i}(S^*_i)\\
\leq \alpha\big((1+O(\eps))\cost_{\mc{I}}(\opt) +O(\eps\Lambda)\Big).
\]
Hence the output $\bigcup_{i=1}^b S_i(k'_i)$ is a 
$(1+O(\eps))\alpha$-approximation, since $\Lambda$ is a constant approximation 
of $\opt$.
\end{proof}

\cref{lem:coarse} implies that if there is a PTAS for coarse instances, we also 
have a PTAS in general. Hence from now on we assume that the given instance 
$\mc{I}$ is coarse w.r.t.~a hierarchical decomposition~$\mc{D}$ of some subset 
$W$ of the vertices of the input graph~$G$, where $\mc{D}$ has bounded scaling 
probability factor, and concise and precise interface sets in $G$, according to 
\cref{lem:decomp} (for some value $\rho>0$ specified later) 

The next step of the algorithm is to compute a new instance $\mc{I}_\mc{D}$ with 
small distortion as given by \cref{lem:struct}. Recall that $\mc{I}_\mc{D}$ is 
obtained from $\mc{I}$ by moving badly cut clients to facilities of $L$. In 
particular, the instance $\mc{I}_\mc{D}$ is also coarse w.r.t.~$\mc{D}$, which 
means that we may run our dynamic program on $\mc{I}_\mc{D}$.

The dynamic program exploits the interface sets of $\mc{D}$ by computing a 
near-optimum ``interface-respecting'' solution to $\mc{I}_\mc{D}$, i.e., a 
solution where clients are connected to facilities through interface points. 
Moreover, for the dynamic program to run in polynomial time it can only estimate 
the distances between interface points and facilities to a certain precision. In 
general, we denote by $\langle x\rangle_i=\min\{(69+\delta)\rho2^i\mid 
\delta\in\mathbb{N} \text{ and } \rho\delta 2^i\geq x\}$ the value of $x$ 
rounded to the next multiple of $\rho 2^i$ and shifted by $69\rho 2^i$. We then 
define the \emph{rounded interface-respecting distance} $\dist'(v,u)$ from a 
vertex $v$ to another vertex $u$ as follows. If $v=u$ then $\dist'(v,u)=0$. 
Otherwise, let $i\geq 1$ be the level of $\mc{D}$ such that there is a part 
$B\in\mc{B}_i$ with $v,u\in B$, and $\mc{D}$ cuts $v$ and $u$ at level $i-1$. We 
let
\[
\dist'(v,u) = \min\big\{\dist(v,p)+\langle\dist(p,u)\rangle_i\mid p\in 
I_B\big\}.
\]
Note that $\dist'(\cdot,\cdot)$ does not necessarily fulfill the triangle 
inequality, and is also not symmetric. We therefore need the bounds of the 
following lemma.

\begin{lemma}\label{lem:dist'}
For any level $i\geq 1$ and vertices $v$ and $u$ that are cut by $\mc{D}$ on 
level $i-1$ we have $\dist'(v,u)\leq\dist(v,u)+138\cdot\rho 2^i$. Let 
$B\in\mc{B}_j$ be the part on some level $j\geq i$ with~$v,u\in B$. For any 
$p\in I_B$ we have $\dist'(v,u)\leq\dist(v,p)+\langle\dist(p,u)\rangle_j$.
\end{lemma}
\begin{proof}
Let $B'\in\mc{B}_i$ be the part on level $i$ containing both $v$ and $u$. By 
\cref{lem:decomp} there is an interface point $p'\in I_{B'}$ such that 
$\dist(v,p') + \dist(p',u) \leq \dist(v,u) + 68\cdot\rho 2^i$. By definition of 
the rounding we also have $\langle\dist(p',u)\rangle_i\leq 
\dist(p',u)+70\cdot\rho 2^i$. Hence $\dist'(v,u)\leq 
\dist(v,p')+\langle\dist(p',u)\rangle_i\leq \dist(v,p')+\dist(p',u)+70\cdot\rho 
2^i\leq \dist(v,u)+138\cdot\rho 2^i$.

The second part is obvious if $j=i$ from the definition of $\dist'(v,u)$. If 
$j\geq i+1$, we use the above bound on $\dist'(v,u)$ together with the additive 
shift of the rounding and the triangle inequality of $\dist(\cdot,\cdot)$ to 
obtain
\[
\dist'(v,u)\leq\dist(v,u) + 138\cdot\rho 2^i
\leq \dist(v,p)+\dist(p,u) +138\cdot\rho 2^{j-1}\\
\leq \dist(v,p)+\langle\dist(p,u)\rangle_j - 69\cdot\rho 2^j +138\cdot\rho 
2^{j-1} = \dist(v,p)+\langle\dist(p,u)\rangle_j. 
\qedhere
\]
\end{proof}

For any non-empty set $S$ of facilities, we define $\dist'(v,S)=\min_{f\in 
S}\{\dist'(v,S)\}$, and for empty sets we let $\dist'(v,\emptyset)=\infty$. 
Analogous to $\cost_{\mc{I}_0}(S)$, for a solution~$S$ to some 
instance~$\mc{I}_0$ we define $\cost'_{\mc{I}_0}(S)$ using $\dist'(\cdot,\cdot)$ 
as
\[
\cost'_{\mc{I}_0}(S) = \sum_{v\in V}\chi_{\mc{I}_0}(v)\cdot\dist'(v,S)^q.
\]

We show the following lemma, which translates between $\cost'_{\mc{I}_\mc{D}}$ 
and $\cost_\mc{I}$, and is implied by the preciseness of the interface sets and 
the fact that $\mc{I}_\mc{D}$ has small distortion (i.e., \cref{lem:struct}). 
Recall that this means that $\nu_{\mc{I}_\mc{D}} \leq \eps \cost_\mc{I}(L)$, and 
there exists a witness solution $\hat S\subseteq F$ that contains the badly cut 
facilities $B_\mc{D}$ and for which
$
\cost_{\mc{I}_\mc{D}}(\hat S) \leq (1+O(\eps))\cost_{\mc{I}}(\opt) + O(\eps) 
\cost_{\mc{I}}(L).
$
Moreover, in the case of \fl, $\hat{S}=\opt\cup B_\mc{D}$ and $\sum_{f \in 
B_{\mc{D}}} w_f \leq \eps\cdot \sum_{f \in L} w_f$. Recall also that the set of 
facilities is the same in $\mc{I}$ and~$\mc{I}_\mc{D}$, i.e., a solution to one 
of these instances is also a solution to the other.

\begin{lemma}\label{lem:int-resp}
Let $\mc{I}$ be an instance of \cluster or \fl with optimum solution $\opt$ and 
approximate solution $L$. Let $\mc{I}_\mc{D}$ be an instance of small distortion 
for some $0<\eps<1/2$, computed from $L$ and a hierarchical 
decomposition~$\mc{D}$ with precise interface sets 
for~$\rho\leq\frac{\eps^{q+4+1/q}} {1104\cdot\sigma(q+1)^q}$ according to 
\cref{lem:decomp}. With probability at least $1-\eps$, for the witness 
solution~$\hat{S}$ of $\mc{I}_\mc{D}$ we have 
$\cost'_{\mc{I}_\mc{D}}(\hat{S})\leq(1+O(\eps))\cost_\mc{I}(\opt) + O(\eps) 
\cost_{\mc{I}}(L)$. Moreover, for any solution~$S$ we have $\cost_\mc{I}(S)\leq 
(1+O(\eps))\cost'_{\mc{I}_\mc{D}}(S)+O(\eps)\cost_\mc{I}(L)$.
\end{lemma}
\begin{proof}
To show the first inequality, we consider the rounded connection costs of 
clients to their closest facility in $\hat{S}$ via some interface point. That 
is, let $c$ be a client of $\mc{I}_\mc{D}$ and let~$f\in\hat{S}$ be its closest 
facility (according to $\dist(\cdot,\cdot)$). If $c\neq f$, there is a level 
$i\geq 1$ for which $\mc{D}$ cuts $c$ and $f$ at level $i-1$. By 
\cref{lem:dist'} we have $\dist'(c,f)\leq\dist(c,f)+138\cdot\rho 2^i$. Also, by 
\cref{lem:struct} we know that $i-1\leq 
\log_2(3L_{\tilde{c}}/\eps+4\opt_{\tilde{c}})+\offset$, where $L_{\tilde{c}}$ 
and $\opt_{\tilde{c}}$ are the respective minimum distances from the original 
position ${\tilde{c}}$ of $c$ to $L$ and $\opt$ in $\mc{I}$. Hence using the 
definitions of $\offset=\log_2(\sigma(q+1)^q/\eps^{q+1})$ and 
$\rho\leq\frac{\eps^{q+4+1/q}} {1104\cdot\sigma(q+1)^q}$ we get
\[
\dist'(c,f) &\leq \dist(c,f) +138\cdot\rho2^i\\
&\leq \dist(c,f) + 138\cdot\rho
2^{\log_2(3L_{\tilde{c}}/\eps+4\opt_{\tilde{c}})+\offset+1}\\
&\leq \dist(c,f)+138\cdot 4(L_{\tilde{c}}+\opt_{\tilde{c}})\cdot 
2^{\offset+1}/\eps \\
&\leq \dist(c,f)+(L_{\tilde{c}}+\opt_{\tilde{c}})\cdot 
1104\cdot\rho\sigma(q+1)^q/\eps^{q+2} \\
&\leq \dist(c,f)+\eps^{2+1/q}(L_{\tilde{c}}+\opt_{\tilde{c}}).
\]
In the other case when $c=f$ we have $\dist'(c,f)=0=\dist(c,f)$, and thus the 
above inequality again holds.

To bound $\dist'(c,f)^q$, we use the bound of \cref{faketriangleineq}, which 
can be stated as $(x+y)^q \leq (1+O(\eps)) x^q + O(1/\eps^q) y^q$ if $q\geq 1$. 
Applying this twice to the bound on $\dist'(c,f)$ above, we get
\[
\dist'(c,f)^q &\leq 
(1+O(\eps))\dist(c,f)^q+ O(\eps^{q+1})(L_{\tilde{c}}+\opt_{\tilde{c}})^q \\
&\leq (1+O(\eps))\dist(c,f)^q+ O(\eps^{q+1}(1+\eps))L_{\tilde{c}}^q 
+O(\eps)\opt_{\tilde{c}}^q \\
&\leq (1+O(\eps))\dist(c,f)^q+ O(\eps)(L_{\tilde{c}}^q+\opt_{\tilde{c}}^q).
\]
To bound $\cost'_{\mc{I}_\mc{D}}(\hat{S})$ using this inequality we define 
$L_{\tilde{v}}=\opt_{\tilde{v}}=0$ for any non-client $v$ of $\mc{I}_\mc{D}$, 
i.e., whenever $\chi_{\mc{I}_\mc{D}}(v)=0$, so that applying the definition of 
$\chi_{\mc{I}_\mc{D}}$ we obtain
\[
\cost'_{\mc{I}_\mc{D}}(\hat{S})&=
\sum_{v\in V}\chi_{\mc{I}_\mc{D}}(v)\cdot\dist'(v,\hat{S})^q\\
&\leq \sum_{v\in V}\chi_{\mc{I}_\mc{D}}(v) 
\left((1+O(\eps))\cdot\dist(v,\hat{S})^q+ 
O(\eps)(L_{\tilde{v}}^q+\opt_{\tilde{v}}^q)\right)\\
&= (1+O(\eps))\cost_{\mc{I}_\mc{D}}(\hat{S}) 
+O(\eps)(\cost_\mc{I}(L)+\cost_\mc{I}(\opt)).
\]
Since $\hat{S}$ is the witness solution of $\mc{I}_\mc{D}$, we know that 
$\cost_{\mc{I}_\mc{D}}(\hat S) \leq (1+O(\eps))\cost_{\mc{I}}(\opt) + O(\eps) 
\cost_{\mc{I}}(L)$ so that also $\cost'_{\mc{I}_\mc{D}}(\hat{S})\leq 
(1+O(\eps))\cost_{\mc{I}}(\opt) + O(\eps) \cost_{\mc{I}}(L)$, as claimed.

For the second inequality of the lemma for any solution $S$, since 
$\mc{I}_\mc{D}$ has small distortion with probability at least $1-\eps$ 
(according to \cref{lem:struct}) we have $\cost_\mc{I}(S) -  (1+2\eps) 
\cost_{\mc{I}_\mc{D}}(S)\leq \nu_{\mc{I}_\mc{D}} \leq \eps \cost_\mc{I}(L)$. 
This immediately implies $\cost_\mc{I}(S)\leq (1+2\eps) 
\cost'_{\mc{I}_\mc{D}}(S)+ \eps \cost_\mc{I}(L)$, since 
$\dist(c,f)\leq\dist'(c,f)$ by the triangle inequality of $\dist(\cdot,\cdot)$ 
and the fact that $\langle x\rangle_i\geq x$ for any $x$.
\end{proof}

The next lemma states the properties of the dynamic program that for any coarse 
instance~$\mc{I}_0$ computes an \emph{optimum rounded interface-respecting 
solution}, which formally is a subset $\opt'$ of facilities that minimizes 
$\cost'_{\mc{I}_0}(\opt')$ with $|\opt'|\leq k$ for \cluster, while for \fl it 
minimizes $\cost'_{\mc{I}_0}(\opt')+\sum_{f\in \opt'}w_f$. This step of the 
algorithm exploits the conciseness of the interface sets and the coarseness of 
the instance to bound the runtime. We prove the following lemma in 
\cref{sec:dp}.

\begin{restatable}{lemma}{lemdp}
\label{lem:dp}
Let $\mc{I}_0$ be an instance of \cluster or \fl that for some~$\eps>0$ is 
coarse w.r.t.~a hierarchical decomposition $\mc{D}$ with concise interface sets 
for some~$1/2\geq\rho>0$ according to \cref{lem:decomp}. An optimum rounded 
interface-respecting solution for~$\mc{I}_0$ can be computed in 
$(nX/\eps)^{(\hw/\rho)^{O(1)}}$ time.
\end{restatable}

We are now ready to put together the above lemmas to prove \cref{thm:main-alg}, 
which we restate here for convenience.

\thmmainalg*

\begin{proof}%
Given an instance of \cluster or \fl we first apply \cref{lem:coarse} to reduce 
to a coarse instance. \cref{lem:coarse} also supplies a hierarchical 
decomposition $\mc{D}$ with the properties given in \cref{lem:decomp}. We use 
this together with a constant approximation $L$ of the coarse instance $\mc{I}$ 
to compute a new instance $\mc{I}_\mc{D}$ with small distortion via 
\cref{lem:struct}. On this instance we apply \cref{lem:dp} to compute an optimum 
rounded interface-respecting solution~$\opt'$ in $(nX/\eps)^{(\hw/\rho)^{O(1)}}$ 
time. Since the facility sets of $\mc{I}$ and~$\mc{I}_\mc{D}$ are the same, we 
may output~$\opt'$ for~$\mc{I}$, which can then be converted into a solution of 
the original non-coarse input instance using \cref{lem:coarse}, while only 
losing a $(1+O(\eps))$-factor in the objective function. Hence it suffices to 
show that $\opt'$ is a $(1+O(\eps))$-approximation for $\mc{I}$ and to bound the 
runtime of the algorithm.

From \cref{lem:int-resp} we get $\cost_\mc{I}(\opt')\leq 
(1+O(\eps))\cost'_{\mc{I}_\mc{D}}(\opt')+O(\eps)\cost_\mc{I}(L)$ by setting 
$\rho\leq\frac{\eps^{q+4+1/q}} {1104\cdot\sigma(q+1)^q}$ with probability 
$1-\eps$. We know that $\cost'_{\mc{I}_\mc{D}}(\opt')\leq 
\cost'_{\mc{I}_\mc{D}}(\hat{S})$ for \cluster, where~$\hat{S}$ is the witness 
solution of $\mc{I}_\mc{D}$. Putting these inequalities together we have 
$\cost_\mc{I}(\opt')\leq \cost'_{\mc{I}_\mc{D}}(\hat{S})+ 
O(\eps)\cost_\mc{I}(\opt)$, since $L$ is a constant approximation to the optimum 
solution $\opt$ to $\mc{I}$. For the same reason, \cref{lem:int-resp} also 
implies that $\cost'_{\mc{I}_\mc{D}}(\hat{S})\leq(1+O(\eps))\cost_\mc{I}(\opt)$, 
which gives $\cost_\mc{I}(\opt')\leq (1+O(\eps))\cost_\mc{I}(\opt)$, i.e., for 
\cluster the solution $\opt'$ is a $(1+O(\eps))$-approximation to $\opt$ with 
probability $1-\eps$.

For \fl we have $\cost'_{\mc{I}_\mc{D}}(\opt')+\sum_{f\in\opt'}w_f\leq 
\cost'_{\mc{I}_\mc{D}}(\hat{S})+\sum_{f\in\hat{S}}w_f$, which by the above 
bounds gives $\cost_\mc{I}(\opt')+\sum_{f\in\opt'}w_f\leq 
(1+O(\eps))\cost_\mc{I}(\opt)+\sum_{f\in\hat{S}}w_f$. In case of \fl we have the 
additional property that $\hat{S}$ is the union of $\opt$ and the badly cut 
clients~$B_\mc{D}$, which implies 
$\sum_{f\in\hat{S}}w_f\leq\sum_{f\in\opt}w_f+\sum_{f\in B_\mc{D}}w_f$. 
Furthermore, we have $\sum_{f\in B_\mc{D}}w_f\leq \eps\cdot\sum_{f\in L}w_f$, 
and hence $\cost_\mc{I}(\opt')+\sum_{f\in\opt}w_f\leq 
(1+O(\eps))\cost_\mc{I}(\opt)+\sum_{f\in\opt}w_f+\eps\cdot\sum_{f\in L}w_f$. 
Again using that $L$ is a constant approximation of $\opt$ we obtain 
$\cost_\mc{I}(\opt')+\sum_{f\in\opt}w_f\leq 
(1+O(\eps))(\cost_\mc{I}(\opt)+\sum_{f\in\opt}w_f)$, i.e., also in this case 
$\opt'$ is a $(1+O(\eps))$-approximation to $\opt$ with probability $1-\eps$.

Next we bound the runtime. According to \cref{lem:int-resp} we need to set 
$\rho\leq\frac{\eps^{q+4+1/q}} {1104\cdot\sigma(q+1)^q}$, while the scaling 
probability factor we obtain from \cref{lem:decomp} is 
$\sigma\leq(\hw\log(1/\rho))^c$ for some constant~$c$. Note that the bound on 
$\rho$ depends on $\sigma$ and vice versa, which means that we need to be 
careful when determining a value for $\rho$ respecting the bound from 
\cref{lem:int-resp}. In particular, substituting the bound for~$\sigma$ in the 
bound for $\rho$ and rearranging, it suffices to set $\rho$ such that 
$\rho\log^c(1/\rho)\leq\frac{\eps^{q+4+1/q}} {1104\cdot(q+1)^q h^c}$. Observe 
that for any value $x>0$ such that $\log^c(1/x^2)\leq 1/x$, setting 
$\rho=\frac{x}{\log^c(1/x^2)}$ implies that 
$\rho\log^c(1/\rho)=x\cdot\frac{\log^c(\log^c(1/x^2)\cdot\frac{1}{x})} {\log^c 
(1/x^2)}\leq x$. Since there exists some constant $c'$ such that 
$\log^c(1/x^2)\leq 1/x$ for any $x\in(0,c']$, for sufficiently small $\eps$ we 
can set $x=\frac{\eps^{q+4+1/q}} {1104\cdot(q+1)^q h^c}$ so that the inequality 
of \cref{lem:int-resp} is fulfilled (note that $\eps$ can be chosen independent 
of $h$ and $q$). Setting $x$ this way also implies 
$\rho=\frac{x}{\log^c(1/x^2)}\geq x^2\geq (\frac{\eps}{hq})^{\Theta(q)}$, and 
thus according to \cref{lem:dp} the runtime of the dynamic program becomes \[ 
(nX/\eps)^{(\hw/\rho)^{O(1)}}\leq (nX)^{(\hw q/\eps )^{O(q)}}. \] All other 
steps of the algorithm run in polynomial time, and so the claimed runtime 
follows.
\end{proof}

\subsection{The dynamic program (proof of \cref{lem:dp})}
\label{sec:dp}

We describe the algorithm for \cluster, and only mention in the end how to 
modify the algorithm to compute a solution for \fl. We develop a dynamic 
program, which follows more or less standard techniques as for instance given in 
\cite{AroraRagRao98,kolliopoulos2007nearly}, but needs to handle the weaker 
preciseness property of the interface.

The solution is computed by a dynamic program recursing on the 
decomposition~$\mc{D}$. Let $W$ be the vertex set that $\mc{D}$ decomposes, and 
which contains all clients and facilities of the coarse instance~$\mc{I}$. 
Roughly speaking, the table of the dynamic program will have an entry for every
part $B\in\mc{B}_i$ of $\mc{D}$ on all levels $i\geq\xi(W)$, for which it will 
estimate the distance from each interface point on all higher levels $j\geq 
i+1$ to the closest facility of the optimum solution. That is, if 
$\tilde{B}\in\mc{B}_j$ is a higher-level part for which $B\subseteq\tilde{B}$, 
then the distances from all interface points $I_{\tilde{B}}$ to facilities of 
the solution in $\tilde{B}$ will be estimated for $B$.

Here the estimation happens in two ways. First off, the distances to facilities 
outside of $B$ have to be guessed. That is, there is an \emph{external} distance 
function $d^+_j$ that assigns a distance to each interface point 
of~$I_{\tilde{B}}$, anticipating the distance from such a point to the closest 
facility of $\tilde{B}$, if this facility lies outside of $B$. In order to 
verify whether the guess was correct, each entry for a part $B$ on level $i$ 
also provides an \emph{internal} distance function~$d^-_j$, which stores the 
distance from each interface point of~$I_{\tilde{B}}$ on level $j\geq i+1$ to 
the closest facility, if the facility is guessed to lie inside of $B$. 

The other way in which distances are estimated concerns the preciseness with 
which they are stored. The distance functions $d^+_j$ and $d^-_j$ will only take 
rounded values $\langle x\rangle_j$ where \mbox{$0<x\leq 2^{j+6}$}, or~$\infty$ 
if no facility at the appropriate distance exists. In particular, if the 
facility of the solution in $\tilde{B}$ that is closest to $p\in I_{\tilde{B}}$ 
lies outside of $B$ then $d^-_j(p)=\infty$, and if it lies inside of $B$ then 
$d^+_j(p)=\infty$. If there is no facility of the solution in $\tilde{B}$ then 
both distance functions $d^+_j$ and $d^-_j$ are set to $\infty$ for all~$p\in 
I_{\tilde{B}}$. Note that this means that at least one of $d^+_j(p)$ and 
$d^-_j(p)$ is always set to $\infty$. Note also that the finite values in the 
domains of the distance functions admit to store the rounded distance to any 
facility in $\tilde{B}$ on level $j$, since the diameter of $\tilde{B}$ is at 
most $2^{j+5}$ by \cref{lem:diam}, and the distance from any $p\in 
I_{\tilde{B}}$ to $\tilde{B}$ is at most $(1+2\rho)\diam(\tilde{B})$ by 
\cref{lem:interface}, i.e., for any $f\in\tilde{B}\cap F$ we have 
$\dist(p,f)\leq (1+2\rho)2^{j+5}\leq 2^{j+6}$ using $\rho\leq 1/2$.

\subparagraph*{Formal definition of the table.}
Let us denote by $I^j_B$ the interface set of the part $\tilde{B}\in\mc{B}_j$ 
on level $j\geq i+1$ containing $B\in\mc{B}_i$, i.e., $I^j_B=I_{\tilde{B}}$. 
Every entry of the dynamic programming table~$T$ is defined by a part 
$B\in\mc{B}_i$ of $\mc{D}$ on a level $i\in\{\xi(W),\ldots,\lambda(W)\}$, and 
two distance functions $d^+_j,d^-_j: I^j_B\to\{\langle x\rangle_j\mid 0<x\leq 
2^{j+6}\}\cup\{\infty\}$ for each $j\in\{i+1,\ldots,\lambda(W)\}$, such that 
$\max\{d^+_j(p),d^-_j(p)\}=\infty$ for all $p\in I^j_B$. Additionally, each 
entry comes with an integer \mbox{$k'\in\{0,\ldots,k\}$}, which is a guess on 
the number of facilities that the optimum solution contains in $B$.

In an entry $T[B,k',(d^+_j,d^-_j)_{j=i+1}^{\lambda(W)}]$ we store the rounded 
interface-respecting cost of connecting the clients of $B$ to facilities that 
adhere to the distance functions. More concretely, let $S\subseteq F\cap B$ be 
any subset of facilities in $B$. We say that $S$ is \emph{compatible} with an 
entry $T[B,k',(d^+_j,d^-_j)_{j=i+1}^{\lambda(W)}]$ if $|S|=k'$, and for any 
$j\geq i+1$ the values of the distance functions for every interface point~$p\in 
I^j_B$ are set to either
\begin{itemize}
\item $d^-_j(p)=\langle\dist(p,S)\rangle_j$ and $d^+_j(p)=\infty$, or
\item $d^+_j(p)\leq\langle\dist(p,S)\rangle_j$ and $d^-_j(p)=\infty$.
\end{itemize}
Note that this means that the distance to a facility outside of $B$ as guessed 
by $d^+$ should not exceed the distance to a facility of $S$ inside of $B$. 
Recall that $\dist(v,\emptyset)=\infty$, and so the empty set $S=\emptyset$ is 
compatible with $T[B,k',(d^+_j,d^-_j)_{j=i+1}^{\lambda(W)}]$ if $k'=0$, and the 
values of all internal distance functions are set to~$\infty$. An entry 
$T[B,k',(d^+_j,d^-_j)_{j=i+1}^{\lambda(W)}]$ for $B\in\mc{B}_i$ should store the 
minimum value $C_B(S)$ over all sets~$S\subseteq F\cap B$ compatible with the 
entry, where $C_B(S)$ is the cost of a solution that opens the facilities of $S$ 
in $B$ but also connects to open facilities outside of $B$ through interface 
points according to the guessed distances given by~$d^+$. Formally, $C_B(S)$ is 
defined as
\[
C_B(S)=\sum_{v\in B} \chi_{\mc{I}_0}(v)\cdot \min\Big\{\dist'(v,S)^q, 
\min_{\substack{j\geq i+1\\ p\in I^j_B}} \big\{(\dist(v,p) 
+d^+_j(p))^q\big\}\Big\}.
\]
If there is no compatible set $S\subseteq F\cap B$ for the entry, then 
$T[B,k',(d^+_j,d^-_j)_{j=i+1}^{\lambda(W)}]=\infty$.

On the highest level $i=\lambda(W)$, there are no distance functions to adhere 
to on levels \mbox{$j\geq i+1$}, and thus any set $S\subseteq W$ of facilities 
is compatible with the entry for $B=W$ and~$k'=|S|$. Furthermore, 
$\cost'_{\mc{I}_0}(S)$ is equal to $C_W(S)$, since $W$ contains all clients and 
facilities of the coarse instance~$\mc{I}_0$. In particular, the entry of $T$ 
for which $k'=k$ and $B=W$, will contain the objective function value of the 
optimum rounded interface-respecting solution to $\mc{I}_0$. Hence if we can 
compute the table $T$ we can also output the optimum rounded 
interface-respecting solution via this entry.

\subparagraph*{Computing the table.}
We begin with a part $B\in\mc{B}_{\xi(W)}$ on the lowest considered 
level~$\xi(W)$, for which we know that $B$ contains at most one facility, as 
$\mc{I}_0$ is coarse. If $B$ contains no facility, then only $S=\emptyset$ can 
be compatible with the entry $T[B,k',(d^+_j,d^-_j)_{j=\xi(W)+1}^{\lambda(W)}]$ 
and computing the value of the entry is straightforward given the definition of 
$C_B(S)$, where all incompatible entries are set to $\infty$. If $B$ contains 
one facility~$f$, then any compatible set $S$ is either empty or only contains 
$f$. We can thus check whether either of the two options is compatible with the 
entry $T[B,k',(d^+_j,d^-_j)_{j=\xi(W)+1}^{\lambda(W)}]$ by checking if $k'$ is 
set to $0$ or~$1$, respectively, and checking that all values of the internal 
distance function are set correctly. Thereafter we can again use the definition 
of $C_B(S)$ to compute the values for both possible sets $S$ and store them in 
the respective compatible entries. All incompatible entries are set to $\infty$.

Now fix a part $B\in\mc{B}_i$ that lies on a level $i>\xi(W)$. We show how to 
recursively compute all entries $T[B,k',(d^+_j,d^-_j)_{j=i+1}^{\lambda(W)}]$ for 
all values~$k'$ and distance functions. By induction we have already computed 
the correct values of all entries of $T$ for parts $B'\in\mc{B}_{i-1}$ where 
$B'\subseteq B$. We order these parts arbitrarily, so that $B'_1,\ldots, B'_b$ 
are the parts of $\mc{B}_{i-1}$ contained in~$B$. We then define an auxiliary 
table $\hat{T}$ that is similar to the table $T$, but should compute the best 
compatible facility set in the union $B'_{\leq\ell}=\bigcup_{h=1}^\ell B'_h$ of 
the first $\ell$ subparts of~$B$. Accordingly, $\hat{T}$ has an entry for each 
union of parts $B'_{\leq\ell}$, each $k'\in\{0,\ldots,k\}$, and distance 
functions $d^+_j,d^-_j: I^j_B\to\{\langle x\rangle_j\mid 0<x\leq 
2^{j+6}\}\cup\{\infty\}$ for each $j\in\{i,\ldots,\lambda(W)\}$, such that 
$\max\{d^+_j(p),d^-_j(p)\}=\infty$ for all $p\in I^j_B$. Here, 
naturally,~$I^i_B=I_B$, i.e., the entry also takes the interface set of $B$ into 
account.

Analogous to before, a set $S\subseteq F\cap B'_{\leq\ell}$ of facilities in the 
union is \emph{compatible} with an entry 
$\hat{T}[B'_{\leq\ell},k',(d^+_j,d^-_j)_{j=i}^{\lambda(W)}]$ if $|S|=k'$, and 
for any $j\geq i$ the values of the distance functions for every interface 
point~$p\in I^j_B$ are set to either
\begin{itemize}
\item $d^-_j(p)=\langle\dist(p,S)\rangle_j$ and $d^+_j(p)=\infty$, or
\item $d^+_j(p)\leq\langle\dist(p,S)\rangle_j$ and $d^-_j(p)=\infty$.
\end{itemize} 
The entry $\hat{T}[B'_{\leq\ell},k',(d^+_j,d^-_j)_{j=i}^{\lambda(W)}]$ should 
store the minimum value of $\hat{C}_{\leq\ell}(S)$ over all compatible 
sets~$S\subseteq F\cap B'_{\leq\ell}$, where $\hat{C}_{\leq\ell}(S)$ is defined 
as
\[
\hat{C}_{\leq\ell}(S)=
\sum_{v\in B'_{\leq\ell}} \chi_{\mc{I}_0}(v)\cdot \min\Big\{\dist'(v,S)^q, 
\min_{\substack{j\geq i\\ p\in I^j_B}} \big\{(\dist(v,p) 
+d^+_j(p))^q\big\}\Big\}.
\] 
If there is no compatible set $S\subseteq F\cap B'_{\leq\ell}$ for the entry, 
then $\hat{T}[B'_{\leq\ell},k',(d^+_j,d^-_j)_{j=i}^{\lambda(W)}]=\infty$.

To compute $T$ using the auxiliary table $\hat{T}$, note that since $B=B'_{\leq 
b}$, any set $S\subseteq F\cap B$ is compatible with the entry 
$T[B,k',(d^+_j,d^-_j)_{j=i+1}^{\lambda(W)}]$ if and only if it is compatible 
with a corresponding entry $\hat{T}[B'_{\leq b},k', 
(d^+_j,d^-_j)_{j=i}^{\lambda(W)}]$ for some internal distance function $d^-_i$ 
on level $i$. Furthermore, if $d^+_i(p)=\infty$ for all $p\in I^i$, then 
$C_B(S)=\hat{C}_{\leq b}(S)$ for such a set $S$. Therefore we can easily compute 
the entry $T[B,k',(d^+_j,d^-_j)_{j=i+1}^{\lambda(W)}]$ from $\hat{T}$ by setting
\[
T[B,k',(d^+_j,d^-_j)_{j=i+1}^{\lambda(W)}]= 
\min_{d^-_i}\Big\{\hat{T}[B'_{\leq b},k', (d^+_j,d^-_j)_{j=i}^{\lambda(W)}] \mid 
\forall p\in I^i_B : d^+_i(p)=\infty\Big\}.
\]

\subparagraph*{Computing the auxiliary table.}
Also computing an entry of $\hat{T}$ for $B'_{\leq 1}$ is easy using the 
entries of $T$ for $B'_1$, since $B'_1=B'_{\leq 1}$ and so (taking the index 
shift of $i$ into account) we have
\[
\hat{T}[B'_{\leq 1},k',(d^+_j,d^-_j)_{j=i}^{\lambda(W)}]= 
T[B'_1,k',(d^+_j,d^-_j)_{j=i}^{\lambda(W)}].
\]

To compute entries of $\hat{T}$ for some $B'_{\leq\ell}$ where $\ell\geq 2$, we 
combine entries of table $T$ for~$B'_\ell$ with entries of table $\hat{T}$ for 
$B'_{\leq\ell-1}$. However we will only combine entries with distance functions 
that imply compatible solutions. More concretely, we say that distance functions 
$(d^+_j,d^-_j)_{j=i}^{\lambda(W)}$ for $B'_{\leq\ell}$, 
$(\delta^+_j,\delta^-_j)_{j=i}^{\lambda(W)}$ for $B'_\ell$, and 
$(\beta^+_j,\beta^-_j)_{j=i}^{\lambda(W)}$ for $B'_{\leq\ell-1}$ are 
\emph{consistent} if for every level $j\geq i$ and $p\in I^j_B$ we have one of
\begin{enumerate}
\item $d^+_j(p)=\delta^+_j(p)=\beta^+_j(p)$ and 
$d^-_j(p)=\delta^-_j(p)=\beta^-_j(p)=\infty$, or
\item $d^-_j(p)=\delta^-_j(p)=\beta^+_j(p)$ and 
$d^+_j(p)=\delta^+_j(p)=\beta^-_j(p)=\infty$, or
\item $d^-_j(p)=\delta^+_j(p)=\beta^-_j(p)$ and 
$d^+_j(p)=\delta^-_j(p)=\beta^+_j(p)=\infty$.
\end{enumerate}

The algorithm now considers all sets of consistent distance functions to compute 
an entry $\hat{T}[B'_{\leq\ell},k',(d^+_j,d^-_j)_{j=i}^{\lambda(W)}]$ for 
$\ell\geq 2$ by setting it to
\[\label{eqn:recursion}
\min\big\{ T[B'_\ell,k'',(\delta^+_j,\delta^-_j)_{j=i}^{\lambda(W)}]+
\hat{T}[B'_{\leq\ell-1},k'-k'',(\beta^+_j,\beta^-_j)_{j=i}^{\lambda(W)}]
\;\mid\\
k''\in\{0,\ldots,k'\} \text{ and } (d^+_j,d^-_j)_{j=i}^{\lambda(W)}, 
(\delta^+_j,\delta^-_j)_{j=i}^{\lambda(W)}, 
(\beta^+_j,\beta^-_j)_{j=i}^{\lambda(W)}
\text{ are consistent} \big\}
\]

We now prove the correctness using two lemmas. The following lemma implies that 
if we only consider consistent distance functions to compute entries 
recursively, then the entries will store values for compatible solutions.

\begin{lemma}\label{lem:correct1}
Let $(d^+_j,d^-_j)_{j=i}^{\lambda(W)}$ for $B'_{\leq\ell}$, 
$(\delta^+_j,\delta^-_j)_{j=i}^{\lambda(W)}$ for $B'_\ell$, and 
$(\beta^+_j,\beta^-_j)_{j=i}^{\lambda(W)}$ for $B'_{\leq\ell-1}$ be consistent 
distance functions, and let $S_1\subseteq B'_\ell\cap F$ and $S_2\subseteq 
B'_{\leq\ell-1}\cap F$ be facility sets. If $S_1$ is compatible with entry 
$T[B'_\ell,|S_1|,(\delta^+_j,\delta^-_j)_{j=i}^{\lambda(W)}]$ and $S_2$ is 
compatible with entry 
$\hat{T}[B'_{\leq\ell-1},|S_2|,(\beta^+_j,\beta^-_j)_{j=i}^{\lambda(W)}]$, then 
the union~$S=S_1\cup S_2$ is compatible with entry 
$\hat{T}[B'_{\leq\ell},|S|,(d^+_j,d^-_j)_{j=i}^{\lambda(W)}]$.
Moreover, $\hat{C}_{\leq\ell}(S)=C_{B'_\ell}(S_1)+\hat{C}_{\leq\ell-1}(S_2)$.
\end{lemma}
\begin{proof}
To prove compatibility of $S$ with the entry 
$\hat{T}[B'_{\leq\ell},|S|,(d^+_j,d^-_j)_{j=i}^{\lambda(W)}]$, it suffices to 
show that the distance functions are set correctly. Fix a level $j\geq i$ and an 
interface point~$p\in I^j_B$. There are three cases to consider, according to 
the definition of consistency of the distance functions. In the first case, all 
three internal distance functions are set to~$\infty$, and all external distance 
functions are set to the same value. In particular, since $S_1$ and $S_2$ are 
compatible with their respective entries, we have 
$d^+_j(p)=\delta^+_j(p)=\beta^+_j(p) 
\leq\min\{\langle\dist(p,S_1)\rangle_j,\langle\dist(p,S_2)\rangle_j\} 
=\langle\dist(p,S)\rangle_j$, as $S=S_1\cup S_2$. In the second case, 
$\beta^-_j(p)=\delta^+_j(p)=\infty$ and so 
$\beta^+_j(p)\leq\langle\dist(p,S_2)\rangle_j$ since $S_2$ is compatible with 
its entry, and $\delta^-_j(p)=\langle\dist(p,S_1)\rangle_j$ since $S_1$ is 
compatible with its entry. Since we also have $\beta^+_j(p)=\delta^-_j(p)$ we 
get $\langle\dist(p,S_1)\rangle_j\leq\langle\dist(p,S_2)\rangle_j$, and hence 
$\langle\dist(p,S)\rangle_j=\langle\dist(p,S_1)\rangle_j$. Consistency 
furthermore implies $d^-_j(p)=\delta^-_j(p)=\langle\dist(p,S)\rangle_j$ and 
$d^+_j(p)=\infty$. The third case is analogous to the second, and therefore $S$ 
is compatible with its entry.

For the second part, we consider the contributions of vertices to the terms 
$\hat{C}_{\leq\ell}(S)$, $C_{B'_\ell}(S_1)$, and $\hat{C}_{\leq\ell-1}(S_2)$, 
and show that they are the same for $\hat{C}_{\leq\ell}(S)$ and for 
$C_{B'_\ell}(S_1)+\hat{C}_{\leq\ell-1}(S_2)$. For this we first fix a vertex 
$v\in B'_{\leq\ell-1}$ for which we need to show that its contribution to 
$\hat{C}_{\leq\ell-1}(S_2)$ and~$\hat{C}_{\leq\ell}(S)$ is the same, as it does 
not contribute to $C_{B'_\ell}(S_1)$. In the following we distinguish the cases 
where its contribution to these two terms is due to a facility or an interface 
point.

The first case is that $\dist'(v,S_2)^q\leq \min_{j\geq i,\; p\in I^j_B} 
\{(\dist(v,p) +\beta^+_j(p))^q\}$, i.e., the contribution of~$v\in 
B'_{\leq\ell-1}$ to $\hat{C}_{\ell-1}(S_2)$ is given by a facility of $S_2$. 
Note that the consistency of the distance functions always implies that 
$\beta^+_j(p)=d^+_j(p)$ or $d^+_j(p)=\infty$ for any level $j\geq i$ and 
interface point~$p\in I^j$, and so $\min_{j\geq i,\; p\in I^j_B} \{(\dist(v,p) 
+\beta^+_j(p))^q\}\leq\min_{j\geq i,\; p\in I^j_B} \{(\dist(v,p) 
+d^+_j(p))^q\}$. At the same time $\dist'(v,S)^q\leq\dist'(v,S_2)^q$ as 
$S_2\subseteq S$. We hence get that $\dist'(v,S)^q\leq \min_{j\geq i,\; p\in 
I^j_B} \{(\dist(v,p) +d^+_j(p))^q\}$, i.e., the contribution of~$v$ to 
$\hat{C}_{\ell}(S)$ is also given by a facility of $S$ in this case. Thus to 
show that the contribution of $v$ to $\hat{C}_{\ell-1}(S_2)$ and 
$\hat{C}_{\ell}(S)$ is the same, we need to show that 
$\dist'(v,S)^q=\dist'(v,S_2)^q$. Note that this is implied if 
$\dist'(v,S)^q\geq \min_{j\geq i,\; p\in I^j_B} \{(\dist(v,p) 
+\beta^+_j(p))^q\}$, since we have $\dist'(v,S)^q\leq\dist'(v,S_2)^q\leq 
\min_{j\geq i,\; p\in I^j_B} \{(\dist(v,p) +\beta^+_j(p))^q\}$. Thus the 
following proves the claim, using that the contribution of~$v$ to 
$\hat{C}_{\ell}(S)$ is given by a facility of $S$.

\begin{claim}\label{clm:contr1}
For $v\in B'_{\leq\ell-1}$, if $\dist'(v,S)^q\leq \min_{j\geq i,\; p\in I^j_B} 
\{(\dist(v,p) +d^+_j(p))^q\}$ then we have\\
$\dist'(v,S)^q\geq \min_{j\geq i,\; p\in I^j_B} \{(\dist(v,p) 
+\beta^+_j(p))^q\}$ or $\dist'(v,S)^q=\dist'(v,S_2)^q$.
\end{claim}
\begin{proof}
Note that $\dist'(v,S)^q\neq\dist'(v,S_2)^q$, means 
$\dist'(v,S)^q<\dist'(v,S_2)^q$ as~$S=S_1\cup S_2$. Hence to prove the claim by 
contradiction, we assume
\begin{enumerate}
 \item $\dist'(v,S)^q\leq \min_{j\geq i,\; p\in I^j_B} \{(\dist(v,p) 
+d^+_j(p))^q\}$,
 \item $\dist'(v,S)^q < \min_{j\geq i,\; p\in I^j_B} \{(\dist(v,p) 
+\beta^+_j(p))^q\}$, and
\item $\dist'(v,S)^q<\dist'(v,S_2)^q$.
\end{enumerate}
The latter inequality implies that the value of $\dist'(v,S)^q$ is obtained for 
some facility $f\in S_1\subseteq B'_{\ell}$. In particular, $v\in 
B'_{\leq\ell-1}$ and $f\in B'_{\ell}$ are cut at level~$i-1$, and so there is an 
interface point $p\in I^i_B$ such that 
$\dist'(v,S)=\dist(v,p)+\langle\dist(p,f)\rangle_i$, and $f$ is the closest 
facility to $p$ in $S$, 
i.e,~$\langle\dist(p,S)\rangle_i=\langle\dist(p,f)\rangle_i$. Using the second 
of the assumed inequalities we get 
\[
(\dist(v,p)+\langle\dist(p,f)\rangle_i)^q=\dist'(v,S)^q<(\dist(v,p) 
+\beta^+_i(p))^q,
\]
and so we can conclude that $\langle\dist(p,f)\rangle_i<\beta^+_i(p)$. 

Using the first assumed inequality, we also get 
\[
(\dist(v,p)+\langle\dist(p,f)\rangle_i)^q=\dist'(v,S)^q\leq 
(\dist(v,p) +d^+_i(p))^q, 
\]
i.e., $\langle\dist(p,f)\rangle_i\leq d^+_i(p)$. Since $S$ is compatible with 
$\hat{T}[B'_{\leq\ell},|S|,(d^+_j,d^-_j)_{j=i}^{\lambda(W)}]$, we have 
$d^+_i(p)=\infty$ or $d^+_i(p)\leq\langle\dist(p,S)\rangle_i$. In the latter 
case we would have 
\[
d^+_i(p)\leq \langle\dist(p,S)\rangle_i =\langle\dist(p,f)\rangle_i 
<\beta^+_i(p), 
\]
which however cannot happen if the 
distance functions are consistent. Thus compatibility of $S$ implies 
$d^+_i(p)=\infty$ and $d^-_i(p)=\langle\dist(p,f)\rangle_i$. In particular, we 
can conclude that $d^-_i(p)$ has a finite value (as $f$ exists) and 
$\beta^+_i(p)$ differs from $d^-_i(p)$. This can only mean that the third of the 
consistency properties applies to $p$ at level~$i$, and so 
$\beta^-_i(p)=d^-_i(p)=\langle\dist(p,f)\rangle_i$.

In particular, also $\beta^-_i(p)$ has a finite value, and using the 
compatibility of $S_2$ with entry 
$\hat{T}[B'_{\leq\ell-1},|S_2|,(\beta^+_j,\beta^-_j)_{j=i}^{\lambda(W)}]$, we 
can conclude that there exists a facility $f'\in S_2\subseteq B'_{\leq\ell-1}$ 
with 
\[
\langle\dist(p,f')\rangle_i=\beta^-_i(p)=\langle\dist(p,f)\rangle_i. 
\]
As $v,f'\in B'_{\leq\ell-1}\subseteq B$ and $B\in\mc{B}_i$, the vertices $v$ and 
$f'$ are cut by $\mc{D}$ on a level below~$i$, and \cref{lem:dist'} implies 
$\dist'(v,f')\leq\dist(v,p)+\langle\dist(p,f')\rangle_i$. But then we have
\[
\dist'(v,S_2)\leq \dist'(v,f')\leq \dist(v,p)+\langle\dist(p,f')\rangle_i
=\dist(v,p)+\langle\dist(p,f)\rangle_i
=\dist'(v,S),
\]
which is a contradiction to the third assumed inequality.
\cqed
\end{proof}

The next case we consider is $\min_{j\geq i,\; p\in I^j_B} \{(\dist(v,p) 
+d^+_j(p))^q\}<\dist'(v,S)^q$, that is, the contribution of~$v$ to 
$\hat{C}_{\ell}(S)$ is given by an interface point. As observed before, we have 
\[
\min_{j\geq i,\; p\in I^j_B} \{(\dist(v,p) +\beta^+_j(p))^q\} &\leq \min_{j\geq 
i,\; p\in I^j_B} \{(\dist(v,p) +d^+_j(p))^q\} \quad \text{ and } \\
\dist'(v,S)^q &\leq \dist'(v,S_2)^q, 
\]
which in this case implies $\min_{j\geq i,\; p\in I^j_B} \{(\dist(v,p) 
+\beta^+_j(p))^q\}<\dist'(v,S_2)^q$, i.e., the contribution of $v$ to 
$\hat{C}_{\ell-1}(S_2)$ is also given by an interface point. Note that it also 
implies $\dist'(v,S)^q> \min_{j\geq i,\; p\in I^j_B} \{(\dist(v,p) 
+\beta^+_j(p))^q\}$, and thus the following claim shows that the contribution 
of $v$ to $\hat{C}_{\ell}(S)$ and $\hat{C}_{\ell-1}(S_2)$ is the same.

\begin{claim}\label{clm:contr2}
For $v\in B'_{\leq\ell-1}$, if $\dist'(v,S)^q > \min_{j\geq i,\; p\in I^j_B} 
\{(\dist(v,p) +\beta^+_j(p))^q\}$ then we have\\
$\min_{j\geq i,\; p\in I^j_B} \{(\dist(v,p) +\beta^+_j(p))^q\}=\min_{j\geq i,\; 
p\in I^j_B} \{(\dist(v,p) +d^+_j(p))^q\}$.
\end{claim}
\begin{proof}
As observed before, the consistency of the distance functions always implies 
$\beta^+_j(p)=d^+_j(p)$ or $d^+_j(p)=\infty$. Thus the contrapositive of the 
claim is
\[
\text{if }& \min_{j\geq i,\; p\in I^j_B} \{(\dist(v,p) +\beta^+_j(p))^q\} 
<\min_{j\geq 
i,\; p\in I^j_B} \{(\dist(v,p) +d^+_j(p))^q\} \\
\text{then we have }&
\dist'(v,S)^q \leq \min_{j\geq i,\; p\in I^j_B} \{(\dist(v,p) 
+\beta^+_j(p))^q\}.
\]
Let $j\geq i$ and $p\in I^j_B$ be the level and interface point for which 
$(\dist(v,p) +\beta^+_j(p))^q$ is minimized. The premise of the contrapositive 
implies $\beta^+_j(p)<d^+_j(p)$ for this particular point $p$ and level $j$. 
This can only be the case if $\beta^+_j(p)<\infty$ and $d^+_j(p)=\infty$. The 
values of $\beta^+_j(p)$ and $d^+_j(p)$ can only differ if the second of the 
consistency properties applies to~$p$ at level $j$, and so 
$\beta^+_j(p)=d^-_j(p)$. Since $\beta^+_j(p)<\infty$, the compatibility of $S$ 
with entry $\hat{T}[B'_{\leq\ell},|S|,(d^+_j,d^-_j)_{j=i}^{\lambda(W)}]$, 
implies $\beta^+_j(p)=d^-_j(p)=\langle\dist(p,S)\rangle_j$.

Now let $f\in S\subseteq B'_{\leq\ell}$ be the facility for which 
$\dist(p,S)=\dist(p,f)$ (which exists as $d^-_j(p)<\infty$). 
As $v,f\in B$ and $B\in\mc{B}_i$, the vertices $v$ and $f$ are cut by $\mc{D}$ 
on a level below $i\leq j$, and thus \cref{lem:dist'} implies $\dist'(v,f)\leq 
\dist(v,p)+\langle\dist(p,f)\rangle_j$.
Then, by definition of $\dist'(v,S)$,
\[
\dist'(v,S)^q\leq (\dist(v,p)+\langle\dist(p,f)\rangle_j)^q
=(\dist(v,p)+\langle\dist(p,S)\rangle_j)^q
=(\dist(v,p)+\beta^+_j(p))^q.
\]
The last term is equal to $\min_{j\geq i,\; p\in I^j_B} \{(\dist(v,p) 
+\beta^+_j(p))^q\}$, which gives the required conclusion.
\cqed
\end{proof}

So far we considered the case when the contribution of $v$ to 
$\hat{C}_{\ell-1}(S_2)$ is given by a facility, or when the contribution of $v$ 
to $\hat{C}_{\ell}(S)$ is given by an interface point. Thus the last case we 
consider is when the contribution of $v$ to $\hat{C}_{\ell-1}(S_2)$ is given by 
an interface point and the contribution of $v$ to $\hat{C}_{\ell}(S)$ is given 
by a facility, i.e., 
\[
\min_{j\geq i,\; p\in I^j_B} \{(\dist(v,p) +\beta^+_j(p))^q\} &<\dist'(v,S_2)^q 
\quad\text{ and }\\
\dist'(v,S)^q &\leq \min_{j\geq i,\; p\in I^j_B} \{(\dist(v,p) +d^+_j(p))^q\}.
\]

We need to show $\dist'(v,S)^q=\min_{j\geq i,\; p\in I^j_B} \{(\dist(v,p) 
+\beta^+_j(p))^q\}$, i.e., the left hand sides of the above two inequalities are 
the same. First assume $\dist'(v,S)^q>\min_{j\geq i,\; p\in I^j_B} \{(\dist(v,p) 
+\beta^+_j(p))^q\}$. Then by \cref{clm:contr2} we obtain 
$\dist'(v,S)^q>\min_{j\geq i,\; p\in I^j_B} \{(\dist(v,p) +d^+_j(p))^q\}$, which 
however contradicts that the contribution of $v$ to $\hat{C}_{\ell}(S)$ is given 
by a facility (i.e., the second of the above inequalities). Hence we must have 
\[ 
\dist'(v,S)^q\leq\min_{j\geq i,\; p\in I^j_B} \{(\dist(v,p) +\beta^+_j(p))^q\}.
\]
It remains to show that also $\dist'(v,S)^q\geq\min_{j\geq i,\; p\in I^j_B} 
\{(\dist(v,p) +\beta^+_j(p))^q\}$. According to \cref{clm:contr1}, if the 
contribution of $v$ to $\hat{C}_{\ell}(S)$ is given by a facility, then we have 
\[
\dist'(v,S)^q &= \dist'(v,S_2)^q\quad\text{ or }\\
\dist'(v,S)^q &\geq \min_{j\geq i,\; p\in I^j_B} \{(\dist(v,p) 
+\beta^+_j(p))^q\}.
\]
In the former case, since the contribution of $v$ to $\hat{C}_{\ell-1}(S_2)$ is 
given by an interface point, we would obtain $\dist'(v,S)^q>\min_{j\geq i,\; 
p\in I^j_B} \{(\dist(v,p) +\beta^+_j(p))^q\}$, for which we saw above that this 
leads to a contradiction via \cref{clm:contr2}. Hence we are left with the other 
implication of \cref{clm:contr1}, which implies that the contribution of $v$ to 
$\hat{C}_{\ell}(S)$ and $\hat{C}_{\ell-1}(S_2)$ is the same.

By analogous arguments, the contribution of any $v\in B'_\ell$ to 
$C_{B'_\ell}(S_1)$ is the same as its contribution to $\hat{C}_{\leq\ell}(S)$. 
Since $B'_\ell$ and $B'_{\leq\ell-1}$ partition the set~$B'_{\leq\ell}$, this 
means that $\hat{C}_{\leq\ell}(S)=C_{B'_\ell}(S_1)+\hat{C}_{\leq\ell-1}(S_2)$, 
as required.
\end{proof}

The next lemma implies that the compatible facility set minimizing 
$\hat{C}_{\leq\ell}(S)$ is considered as a solution when recursing over 
consistent distance functions.

\begin{lemma}\label{lem:correct2}
Let $S\subseteq B'_{\leq\ell}\cap F$ be a facility set of $B'_{\leq\ell}$ that 
is compatible with 
$\hat{T}[B'_{\leq\ell},|S|,(d^+_j,d^-_j)_{j=i}^{\lambda(W)}]$, and let 
$S_1=S\cap B'_\ell$ and $S_2=S\cap B'_{\leq\ell-1}$. Then there exist distance 
functions $(\delta^+_j,\delta^-_j)_{j=i}^{\lambda(W)}$ for $B'_\ell$, and 
$(\beta^+_j,\beta^-_j)_{j=i}^{\lambda(W)}$ for $B'_{\leq\ell-1}$ such that
\begin{itemize}
\item $(d^+_j,d^-_j)_{j=i}^{\lambda(W)}$, 
$(\delta^+_j,\delta^-_j)_{j=i}^{\lambda(W)}$, and 
$(\beta^+_j,\beta^-_j)_{j=i}^{\lambda(W)}$ are consistent, and
\item the set $S_1$ is compatible with entry 
$T[B'_\ell,|S_1|,(\delta^+_j,\delta^-_j)_{j=i}^{\lambda(W)}]$ and $S_2$ is 
compatible with entry 
$\hat{T}[B'_{\leq\ell-1},|S_2|,(\beta^+_j,\beta^-_j)_{j=i}^{\lambda(W)}]$.
\end{itemize}
\end{lemma}
\begin{proof}
Consider any interface point $p\in I^j_B$ on some level $j\geq i$. Since $S$ is 
compatible with entry 
$\hat{T}[B'_{\leq\ell},|S|,(d^+_j,d^-_j)_{j=i}^{\lambda(W)}]$, we have exactly 
one of the following: 
\begin{itemize}
 \item $d^+_j(p)=\infty$ and $d^-_j(p)=\langle\dist(p,S)\rangle_j$, or 
 \item $d^+_j(p)\leq\langle\dist(p,S)\rangle_j$ and $d^-_j(p)=\infty$.
\end{itemize}
Note that $S=S_1\cup S_2$ implies that $\langle\dist(p,S)\rangle_j= 
\min\{\langle\dist(p,S_1)\rangle_j,\langle\dist(p, S_2)\rangle_j\}$. Hence if 
$d^+_j(p)\leq\langle\dist(p,S)\rangle_j$ we may set 
$\beta^+_j(p)=\delta^+_j(p)=d^+_j(p)$ and $\beta^-_j(p)=\delta^-_j(p)=\infty$, 
and obtain the first case of the consistency properties. Observe that this also 
implies the compatibility property of $S_1$ and $S_2$ for $p$.

Now assume that $d^+_j(p)=\infty$ and $\langle\dist(p,S_1)\rangle_j\leq 
\langle\dist(p,S_2)\rangle_j$. Then we may set 
$\delta^-_j(p)=\beta^+_j(p)=d^-_j(p)$ and $\delta^+_j(p)=\beta^-_j(p)=\infty$. 
Since $d^-_j(p)=\langle\dist(p,S)\rangle_j=\langle\dist(p, S_1)\rangle_j$, this 
gives the second case of the consistency properties. Again, this also implies 
the compatibility property of $S_1$ and $S_2$ for $p$.

The remaining case when $d^+_j(p)=\infty$ and $\langle\dist(p,S_1)\rangle_j>
\langle\dist(p,S_2)\rangle_j$ is analogous. Here we may set 
$\delta^+_j(p)=\beta^-_j(p)=d^-_j(p)$ and $\delta^-_j(p)=\beta^+_j(p)=\infty$. 
Because $d^-_j(p)=\langle\dist(p,S)\rangle_j=\langle\dist(p, S_2)\rangle_j$, 
this gives the third case of the consistency properties, and also implies the 
compatibility property of $S_1$ and $S_2$ for $p$, which concludes the proof.
\end{proof}

We now argue that the algorithm sets the value of 
$\hat{T}[B'_{\leq\ell},k',(d^+_j,d^-_j)_{j=i}^{\lambda(W)}]$ correctly via the 
recursion given in~\eqref{eqn:recursion}. By induction, for all 
$k''\in\{0,\ldots,k'\}$, any entry 
$T[B'_\ell,k'',(\delta^+_j,\delta^-_j)_{j=i}^{\lambda(W)}]$ stores the minimum 
value of $C_{B'_\ell}(\tilde S_1)$ for some compatible set $\tilde S_1\subseteq 
F\cap B'_\ell$, and any entry \mbox{$\hat{T}[B'_{\leq\ell-1},k'-k'', 
(\beta^+_j,\beta^-_j)_{j=i}^{\lambda(W)}]$} stores the minimum value of 
$\hat{C}_{\leq\ell-1}(\tilde S_2)$ for some compatible set $\tilde S_2\subseteq 
F\cap B'_{\leq\ell-1}$. Thus by \cref{lem:correct1} and the 
recursion of~\eqref{eqn:recursion} we have 
$\hat{T}[B'_{\leq\ell},k',(d^+_j,d^-_j)_{j=i}^{\lambda(W)}]=\hat{C}_{\leq\ell}( 
\tilde S)$ for some compatible set $\tilde S$. Now consider the set $S\subseteq 
F\cap B'_{\leq\ell}$ that is compatible with 
$\hat{T}[B'_{\leq\ell},k',(d^+_j,d^-_j)_{j=i}^{\lambda(W)}]$ and 
minimizes~$\hat{C}_{\leq\ell}(S)$. For $S_1=S\cap B'_\ell$ and $S_2=S\cap 
B'_{\leq\ell-1}$, \cref{lem:correct2}  implies that there are entries 
$T[B'_\ell,|S_1|,(\delta^+_j,\delta^-_j)_{j=i}^{\lambda(W)}]$ and 
$\hat{T}[B'_{\leq\ell-1},|S_2|,(\beta^+_j,\beta^-_j)_{j=i}^{\lambda(W)}]$ that 
store values at most $C_{B'_\ell}(S_1)$ and $\hat{C}_{\leq\ell-1}(S_2)$, 
respectively, and the sum of these two values is considered in the 
recursion~\eqref{eqn:recursion}. Hence $\hat{C}_{\leq\ell}(\tilde 
S)\leq\hat{C}_{\leq\ell}(S)$, but since $S$ minimizes $\hat{C}_{\leq\ell}(S)$ 
the latter is an equality and we get 
$\hat{T}[B'_{\leq\ell},k',(d^+_j,d^-_j)_{j=i}^{\lambda(W)}]= 
\hat{C}_{\leq\ell}(S)$, as required.

\subparagraph*{Bounding the runtime.}
To bound the size of the tables $T$ and $\hat{T}$, note that since there are 
$\lambda(W)-\xi(W)+1\leq 2\log_2(nX/\eps)+2$ considered levels $i$, and each 
level $\mc{B}_i$ of $\mc{D}$ is a partition of~$W$ where $|W|\leq n$, there are 
at most $O(n\log(nX/\eps))$ parts $B$ considered by $T$ in total. The other 
table $\hat{T}$ considers the same number of parts, since a set $B'_{\leq\ell}$ 
can be uniquely mapped to the part $B'_\ell$. The number of possible values for 
$k'$ is $k+1=O(n)$. The domain $\{\langle x\rangle_j\mid 0<x\leq 
2^{j+6}\}\cup\{\infty\}$ of a distance function for level $j$ has at most 
$\lceil 2^{j+6}/(\rho 2^j)\rceil+1=O(1/\rho)$ values, since $\langle x\rangle_j$ 
rounds a value to a multiple of $\rho 2^j$. The conciseness of the interface 
sets means that $|I^j_B|\leq (\hw/\rho)^{O(1)}$ according to \cref{lem:decomp}. 
Hence there are at most $O(1/\rho)^{(\hw/\rho)^{O(1)}}=2^{(\hw/\rho)^{O(1)}}$ 
possible distance functions. Since each entry of the table stores two distance 
functions for each of at most $2\log_2(nX/\eps)+2$ levels, the total number of 
entries of $T$ and $\hat{T}$ is at most
\[
O(n\log(nX/\eps)) \cdot n \cdot (2^{(\hw/\rho)^{O(1)}})^{O(\log(nX/\eps))}
= (nX/\eps)^{(\hw/\rho)^{O(1)}}.
\]

Computing an entry of a table is dominated by~\eqref{eqn:recursion}. Going 
through all values $k'\leq n$ and all possible consistent distance functions to 
compute~\eqref{eqn:recursion}, takes $n\cdot 2^{(\hw/\rho)^{O(1)}}$ time, as 
there are~$2^{(\hw/\rho)^{O(1)}}$ possible distance functions. Hence the total 
runtime is $(nX/\eps)^{(\hw/\rho)^{O(1)}}$, proving \cref{lem:dp} for \cluster.

\subparagraph*{The \fl problem.}
To compute an optimum rounded interface-respecting solution to \fl, the tables
$T$ and $\hat{T}$ can ignore the number of open facilities~$k'$, i.e., they 
have respective entries $T[B,(d^+_j,d^-_j)_{j=i+1}^{\lambda(W)}]$ and 
$\hat{T}[B'_{\leq\ell},(d^+_j,d^-_j)_{j=i}^{\lambda(W)}]$. Accordingly, 
compatibility of facility sets with entries is defined as before, but ignoring 
the sizes of the sets. The value stored in each entry now also takes the opening 
costs of facilities into account. That is, for any set of facilities $S\subseteq 
F\cap B$ in a part $B$ we define
\[
C_B(S)=\sum_{v\in B} \chi_{\mc{I}_0}(v)\cdot \min\Big\{\dist'(v,S)^q, 
\min_{\substack{j\geq i+1\\ p\in I^j_B}} \big\{(\dist(v,p) 
+d^+_j(p))^q\big\}\Big\} +\sum_{f\in S}w_f,
\]
and an entry $T[B,(d^+_j,d^-_j)_{j=i+1}^{\lambda(W)}]$ stores the minimum value 
of $C_B(S)$ over all sets $S$ compatible with the entry, or $\infty$ if no such 
set exists. For $S\subseteq F\cap B'_{\leq\ell}$ in a union of 
subparts~$B'_{\leq\ell}$ we define
\[
\hat{C}_{\leq\ell}(S)=
\sum_{v\in B'_{\leq\ell}} \chi_{\mc{I}_0}(v)\cdot \min\Big\{\dist'(v,S)^q, 
\min_{\substack{j\geq i\\ p\in I^j_B}} \big\{(\dist(v,p) 
+d^+_j(p))^q\big\}\Big\} +\sum_{f\in S}w_f,
\] 
and an entry $\hat{T}[B'_{\leq\ell},(d^+_j,d^-_j)_{j=i}^{\lambda(W)}]$ stores 
the minimum value of $\hat{C}_{\leq\ell}(S)$ over all sets $S$ compatible with 
the entry, or $\infty$ if no such set exists.

The entries of the tables can be computed in the same manner as before, but 
ignoring the set sizes. In particular, the most involved recursion becomes
\[
\hat{T}[B'_{\leq\ell},(d^+_j,d^-_j)_{j=i}^{\lambda(W)}]=
\min\big\{ T[B'_\ell,(\delta^+_j,\delta^-_j)_{j=i}^{\lambda(W)}]+
\hat{T}[B'_{\leq\ell-1},(\beta^+_j,\beta^-_j)_{j=i}^{\lambda(W)}] \;\mid\\
 (d^+_j,d^-_j)_{j=i}^{\lambda(W)}, (\delta^+_j,\delta^-_j)_{j=i}^{\lambda(W)}, 
(\beta^+_j,\beta^-_j)_{j=i}^{\lambda(W)}
\text{ are consistent} \big\}.
\]

Note that if $S_1=B'_\ell\cap F$ and $S_2=B'_{\leq\ell-1}\cap F$ then these two 
sets are disjoint, and so $\sum_{f\in S}w_f=\sum_{f\in S_1}w_f+\sum_{f\in 
S_2}w_f$ for the union $S=S_1\cup S_2$. Hence when proving 
$\hat{C}_{\leq\ell}(S)=C_{B'_\ell}(S_1)+\hat{C}_{\leq\ell-1}(S_2)$ for 
\cref{lem:correct1}, we can ignore the facility opening costs, and the proof 
remains the same as before. All other arguments carry over, and thus an optimum 
rounded interface-respecting solution for an instance of \fl can also be 
computed in $(nX/\eps)^{(\hw/\rho)^{O(1)}}$ time.

\section{Hardness for graphs of highway dimension 1}
\label{sec:hard}

In this section we prove our hardness result, which we restate here for 
convenience.

\thmhardness*

For both \cluster and \fl we present the same reduction from the NP-hard 
satisfiability problem (SAT), in which a boolean formula $\varphi$ in 
conjunctive normal form is given, and a satisfying assignment of its variables 
needs to be found.

For a given SAT formula $\varphi$ with $k$ variables and $\ell$ clauses we 
construct a graph $G_\varphi$ as follows. For each variable $x$ we introduce a 
path $P_x=(t_x,u_x,f_x)$ with two edges of length~$1$ each. The two endpoints 
$t_x$ and $f_x$ are facilities of $F$ and the additional vertex $u_x$ is a 
client, i.e., $\chi(u_x)=1$. For each clause $C_i$, where $i\in[\ell]$, we 
introduce a vertex $v_i$ and add the edge~$v_it_x$ for each variable $x$ such 
that $C_i$ contains $x$ as a positive literal, and we add the edge $v_if_x$ for 
each $x$ for which $C_i$ contains $x$ as a negative literal. Every edge incident 
to $v_i$ has length $(11c)^i$ for the constant $c>4$ due to \cref{dfn:hd}, and 
$v_i$ is also a client, i.e., $\chi(v_i)=1$. In case of \fl, every facility 
$f\in F$ has cost $w_f=1$, i.e., we construct an instance of the 
uniform version of the problem.

\begin{lemma}\label{lem:hd1}
The constructed graph $G_\varphi$ has highway dimension $1$.
\end{lemma}
\begin{proof}
Fix a scale $r>0$ and let $i=\lfloor \log_{11c} (r/5) +1 \rfloor$. Note that 
$\beta_w(cr)$ cannot contain any edge incident to a vertex $v_j$ for $j\geq 
i+1$, since the length of every such edge is $(11c)^j\geq 11cr/5>2cr$ and the 
diameter of $\beta_w(cr)$ is at most $2cr$. Thus if $\beta_w(cr)$ contains a 
vertex $v_j$ for $j\geq i+1$, then $\beta_w(cr)$ contains only $v_j$, and there 
is nothing to prove. Note also that any path in $\beta_w(cr)$ that does not use 
$v_i$ has length at most $2+\sum_{j=1}^{i-1} (2(11c)^j+2)$, since any such path 
can contain at most two edges incident to a vertex $v_j$ and the paths $P_x$ of 
length $2$ are connected only through edges incident to vertices~$v_j$. The 
length of such a path is thus strictly shorter than
\[
2+2\left(\frac{(11c)^i}{11c-1}-1\right)+2i \leq 5(11c)^{i-1} \leq r,
\]
where the first inequality holds since $i\geq 1$ and $c>4$. Hence the only 
paths that need to be hit by hubs on scale $r$ are those passing through $v_i$, 
which can clearly be done using only one hub, namely~$v_i$.
\end{proof}

To finish the reduction for \cluster, we claim that there is a satisfying 
assignment for $\varphi$ if and only if there is a solution for 
$G_\varphi$ with cost at most $k+\sum_{i=1}^\ell (11c)^{iq}$. If there is a 
satisfying assignment for $\varphi$ we open each facility $t_x$ for 
variables~$x$ that are set to true, and we open each facility $f_x$ for 
variables $x$ that are set to false. This opens exactly $k$ facilities and the 
cost of the solution is $k+\sum_{i=1}^\ell (11c)^{iq}$, since each of the $k$ 
vertices $u_x$ is assigned to either $t_x$ or~$f_x$ at distance~$1$, and vertex 
$v_i$ is assigned to a vertex $t_x$ or $f_x$ at distance $(11c)^i$ that 
corresponds to a literal of~$C_i$ that is true. 

Conversely, assume there is a solution to \cluster of cost at most 
$k+\sum_{i=1}^\ell (11c)^{iq}$ in~$G_\varphi$. Note that the minimum distance 
from any $u_x$ to a facility is $1$, while the minimum distance from any $v_i$ 
to a facility is~$(11c)^i$. Thus any solution must have cost at least 
$k+\sum_{i=1}^\ell (11c)^{iq}$, so that the assumed solution must open a 
facility at minimum distance for each client of $G_\varphi$. In particular, for 
each variable $x$, at least one of the facilities $t_x$ and $f_x$ is opened by 
the solution. Moreover, as only $k$ facilities can be opened and there are $k$ 
variables, exactly one of $t_x$ and $f_x$ is opened for each $x$. Thus the 
\cluster solution in $G_\varphi$ can be interpreted as an assignment for 
$\varphi$, where we set a variable $x$ to true if $t_x$ is opened, and we set it 
to false if $f_x$ is opened. Since also for each $v_i$ the solution opens a 
facility at minimum distance, there must be a variable in $C_i$ that is set so 
that its literal in $C_i$ is true, i.e., the assignment satisfies $\varphi$. 
Thus due to the above lemma bounding the highway dimension of $G_\varphi$, we 
obtain the \cref{thm:hardness} for \cluster.

For \fl we claim that there is a satisfying assignment for $\varphi$ if and only 
if there is a solution for $G_\varphi$ of cost at most $2k+\sum_{i=1}^\ell 
(11c)^{iq}$. In fact the arguments are exactly the same as for \cluster 
above: if there is a satisfying assignment then a solution for \fl of cost 
$2k+\sum_{i=1}^\ell (11c)^{iq}$ exists, by opening the $k$ facilities 
corresponding to the assignment of cost~$1$ each. Conversely, any solution has 
cost at least $k+\sum_{i=1}^\ell (11c)^{iq}$ due to the edge lengths, and at 
least $k$ facilities need to be opened, one for each variable gadget. This gives 
a minimum cost of $2k+\sum_{i=1}^\ell (11c)^{iq}$, and any such solution 
corresponds to a satisfying assignment of $\varphi$. This proves 
\cref{thm:hardness} for uniform \fl.

\printbibliography

\newpage
\appendix
\section{Highway dimension vs doubling dimension}\label{ap:hwdef}

We discuss here the relationship between low doubling and low highway dimension 
metrics. First off, in a follow-up paper to~\cite{abraham2010highway}, 
\citet{abraham2016highway} define a version of the highway dimension, which 
implies that the graphs also have bounded doubling dimension. Hence for this 
definition, the algorithm of \citet{DBLP:conf/focs/SaulpicCF19} is already a 
very efficient approximation scheme in metrics of low highway dimension. 
\cref{dfn:hd} on the other hand implies metrics of large doubling dimension as 
noted by \citet{abraham2010highway}: a star has highway dimension~$1$ (by using 
the center vertex to hit all paths), but its doubling dimension is unbounded. 
While it may be reasonable to assume that road networks have low doubling 
dimension (which are the main concern in the works of 
\citet{abraham2016highway,abraham2010highway,abraham2011vc}), there are metrics 
modeling transportation networks, for which it can be argued that the doubling 
dimension is large, while the highway dimension should be small, and thus rather 
adhere to \cref{dfn:hd}: in networks arising from public transportation, longer 
connections are serviced by larger and and sparser stations (such as train 
stations and airports). More concretely, the so-called hub-and-spoke networks 
that can typically be seen in air traffic networks is much closer to a star-like 
network and is unlikely to have small doubling dimension, while still having 
small highway dimension. Thus in these examples it is reasonable to assume that 
the doubling dimension is a lot larger than the highway dimension.

For further discussions on different definitions of the highway dimension we 
refer to the work of \citet{blum2019hierarchy} and Section~9 of 
\citet{FeldmannFKP15}.

\begin{figure}[b!]
\centering
\includegraphics[width=\textwidth]{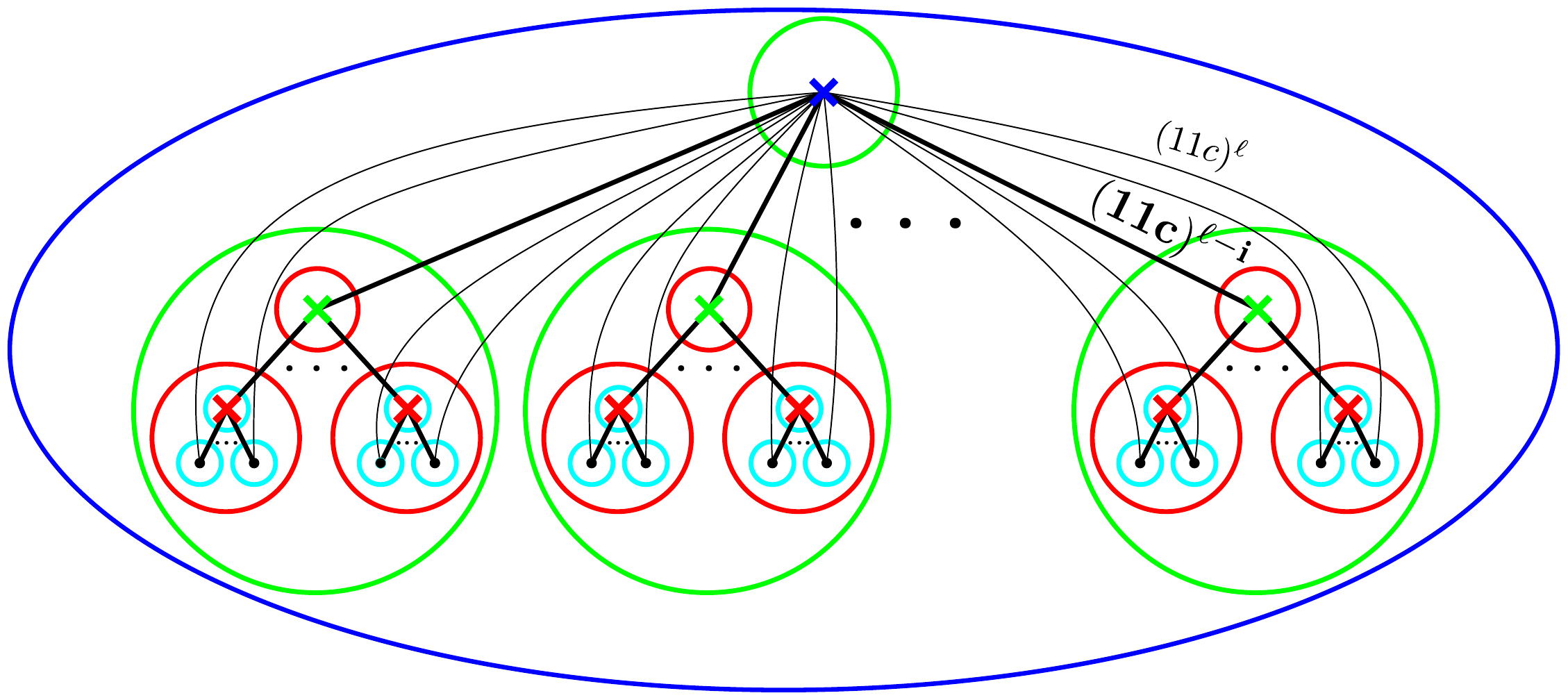}
\caption{An example of a low highway dimension graph without an analogue to 
portal-respecting paths.}
\label{fig:interface}
\end{figure}

Low doubling and low highway dimension metrics still have some similarities, 
which are exploited for the results of this and other papers. However, as 
pointed out in the introduction, a major difference is the non-existence of an 
analogue to portal-respecting paths in low highway dimension metrics. We 
illustrate this by the example given in \cref{fig:interface}: the given graph 
contains a rooted tree (thick black edges) of height $\ell$ (in the picture, 
$\ell=3$) for which each internal vertex at distance~$i$ from the root (topmost 
vertex) is connected to~$\Delta$ children via edges of length~$(11c)^{\ell-i}$, 
where $c>4$ is the constant from \cref{dfn:hd}. Additionally, each leaf of the 
tree is adjacent to the root via a (thin black) edge of length~$(11c)^\ell$. The 
resulting graph has doubling dimension $\log_2(\Delta)$, while the highway 
dimension is~$1$ (using a similar proof as for \cref{lem:hd1}). The circles 
represent towns of the town decomposition, each colour being a different level 
of the decomposition. For each town~$T$ its set~$X_T$ contains one hub (crosses) 
with the same colour (e.g., each green town contains a green hub connecting its 
red child towns), except for the towns on the lowest level (light blue) where 
$X_T=\emptyset$.

In this example any portal set according to \cref{lem:talwar-decomp} will not be 
of constant size, since the doubling dimension is unbounded. Thus we want to 
exploit the towns and their hub sets instead. But then no analogue to 
near-optimal portal-respecting paths exists: connecting a leaf to the root 
through hubs of increasing levels results in the path through the tree (thick 
edges). This path has length $\sum_{i=0}^{\ell-1} 
(11c)^{\ell-i}>(11c)^\ell+(11c)^{\ell-1}$, while using the direct (thin) edge is 
shorter by a factor of more than~$1+\frac{1}{11c}$. Thus for sufficiently 
small~$\eps$, to obtain a near-optimal connection we must use the direct edge, 
i.e., the leaf and root need to be connected using the interface point (the dark 
blue hub) of the town containing them.

\section{Proof sketch of \cref{lem:townproperties}}
\label{ap:proof}

\lemtownproperties*

\begin{proof}[Proof sketch.]
The first inequality follows immediately from \cite[Lemma~3.2]{FeldmannFKP15}. 
For the second inequality, we note that in \cite{FeldmannFKP15} the towns of the 
town decomposition~$\mc{T}$ are defined with respect to exponentially growing 
values $r_i=(c/4)^i$ where $c>4$ is the constant of \cref{dfn:hd}. Here the 
index $i\in\mathbb{N}_0$ is called a \emph{level}, but these levels behave quite 
differently from the levels of a hierarchical decomposition, which is why we 
refrained from introducing levels of town decompositions in this paper. In 
particular, the child towns of a town of level~$i$ might not be from level 
$i-1$, but can be from other levels as well. By \cite[Lemma~3.3]{FeldmannFKP15} 
however, the level of a child town of any town of level $i$ is at most $i-1$. If 
the level of a town $T$ is $i$, then by \cite[Lemma~3.2]{FeldmannFKP15} we have 
$\diam(T)\leq r_i$. Now let~$i$ be the level for which $\diam(T_0)\in 
(r_{i-1},r_i]$ for a given town $T_0$. By the above properties, the level of a 
$g$\textsuperscript{th}-generation descendant $T_g$ of $T_0$ is at most $i-g$, 
and so $\diam(T_g)\leq r_{i-g}$. Since $\diam(T_0)\geq r_{i-1}$ this implies the 
required bound $\diam(T_g)\leq \diam(T_0)/2^{g-1}$, if we set $c=8$.
\end{proof}

\end{document}